\documentclass[11pt]{article}
\usepackage[letterpaper, margin=1in]{geometry}

\usepackage{booktabs} 
\usepackage{amsmath,amsfonts,amssymb,bbm,amsthm, bm}
\usepackage[numbers]{natbib}

\usepackage{graphicx}

\usepackage[pdfencoding=auto]{hyperref}

\usepackage{algorithm,algorithmic}

\usepackage{xspace,xcolor}

\newtheorem{theorem}{Theorem}[section]
\newtheorem{lemma}[theorem]{Lemma}
\newtheorem{claim}[theorem]{Claim}
\newtheorem{corollary}[theorem]{Corollary}
\newtheorem{proposition}[theorem]{Proposition}

\theoremstyle{definition}
\newtheorem{definition}[theorem]{Definition}
\newtheorem{remark}[theorem]{Remark}

\newcommand{\N}{{\mathbb N}}

\newcommand{\Z}{{\mathbb Z}}
\newcommand{\R}{{\mathbb R}}

\newcommand{\calG}{{\mathcal G}}

\newcommand{\infmax}{{\sc InfMax}\xspace}

\DeclareMathOperator*{\E}{\mathbb{E}}

\DeclareMathOperator*{\argmax}{argmax}

\DeclareMathOperator*{\po}{Po}

\begin{document}
\title{Think Globally, Act Locally:  On the Optimal Seeding for Nonsubmodular Influence Maximization\footnote{A short version of this paper is appeared in RANDOM'19. Grant Schoenebeck, Biaoshuai Tao, and Fang-Yi Yu are pleased to acknowledge the support of National Science Foundation AitF \#1535912 and CAREER \#1452915.}}
\author{Grant Schoenebeck\thanks{University of Michigan, School of Information, schoeneb@umich.edu}\and Biaoshuai Tao\thanks{University of Michigan, Computer Science and Engineering Division, bstao@umich.edu}\and Fang-Yi Yu\thanks{University of Michigan, School of Information, fayu@umich.edu}}

\date{}
\maketitle
\begin{abstract}
We study the $r$-complex contagion influence maximization problem. In the influence maximization problem, one chooses a fixed number of initial seeds in a social network to  maximize the spread of their influence. In the $r$-complex contagion model, each uninfected vertex in the network becomes infected if it has at least $r$ infected neighbors.

In this paper, we focus on a random graph model named the \emph{stochastic hierarchical blockmodel}, which is a special case of the well-studied \emph{stochastic blockmodel}.
When the graph is not exceptionally sparse, in particular, when each edge appears with probability $\omega\left(n^{-(1+1/r)}\right)$, under certain mild assumptions, we prove that the optimal seeding strategy is to put all the seeds in a single community.
This matches the intuition that in a nonsubmodular cascade model placing seeds near each other creates synergy.  However, it sharply contrasts with the intuition for submodular cascade models (e.g., the independent cascade model and the linear threshold model) in which nearby seeds tend to erode each others' effects.
Our key technique is a novel time-asynchronized coupling of four cascade processes.

Finally, we show that this observation yields a polynomial time dynamic programming algorithm which outputs optimal seeds if each edge appears with a probability either in $\omega\left(n^{-(1+1/r)}\right)$ or in $o\left(n^{-2}\right)$.
\end{abstract}

\section{Introduction}
\label{sect:introduction}
A \emph{cascade}, or a \emph{contagion}\footnote{As is common in the literature, we use these terms interchangeably.}, is a fundamental process on social networks: starting with some seed agents, the infection then spreads to their neighbors.  A natural question known as influence maximization~\cite{bass1969new, brown1987social, goldenberg2001using, mahajan1991new} asks how to place a fixed number of initial seeds to maximize the spread of the resulting cascade.  For example, which students can most effectively be enrolled in an intervention to decrease student conflict at a school~\cite{paluck2016changing}?

Influence maximization is extensively studied when the contagion process is submodular (a node's marginal probability of becoming infected after a new neighbor is infected decreases when the number of previously infected neighbors increases ~\cite{kempe2003maximizing}).  However, many examples of nonsubmodular contagions have been reported, including pricey technology innovations, the change of social behaviors, the decision to participate in a migration, etc~\cite{coleman1966medical, macdonald1964chain,RomeroMK11,BackstromHKL06,LeskovecAH06}.  In this case, a node's marginal influence may increase in the presence of other nodes---creating a kind of synergy.

\paragraph*{Network structure and seed placement}
We address this lack of understanding for nonsubmodular influence maximization by characterizing the optimal seed placement for certain settings which we will remark on shortly.
In these settings, the optimal seeding strategy is to put all the seeds near each other.
This is significantly different than in the submodular setting, where the optimal solutions tend to spread out the seeds, lest they erode each others' influence.
We demonstrate this in Sect.~\ref{sect:submodular} by presenting an example of submodular influence maximization where the optimal seeding strategy is to spread out the seeds.

This formally captures the intuition, as presented by Angell and Schoenebeck~\cite{angell2017don}, that it is better to target  one market to saturation first (act locally) and then to allow the success in this initial market to drive broader success (think globally) rather than to initially attempt a scattershot approach (act globally).
It is also underscores the need to understand the particular nature of a contagion before blindly applying influence maximization tools.

We consider a well-known nonsubmodular diffusion model which is also the most extreme one (in terms of nonsubmodularity), the $r$-complex contagion~\cite{Granovetter78,centola2007complex,chalupa1979bootstrap,essam1980percolation} (a node is infected if and only if at least $r$ of its neighbors are infected, also known as \emph{bootstrap percolation}) when $r \geq 2$.

We consider networks formed by the \emph{stochastic hierarchical blockmodel}~\cite{schoenebeck2017beyond,schoenebeck2019beyond} which is a special case of the stochastic blockmodel~\cite{dimaggio1986structural, holland1983stochastic, white1976social} equipped with a hierarchical structure.
Vertices are partitioned into $m$ blocks.
The blocks are arranged in a hierarchical structure which represents blocks merging to form larger and larger blocks (communities).
The probability of an edge's presence between two vertices is based solely on smallest block to which both the vertices belong.  This model captures the intuitive hierarchical structure which is also observed in many real-world networks~\cite{girvan2002community,clauset2008hierarchical}.
The stochastic hierarchical blockmodel is rather general and captures other well-studied models (e.g. Erd\H{o}s-R\'{e}nyi random graphs, and the planted community model) as special cases.

\noindent\textbf{Result 1:} We first prove that, for the influence maximization problem on the stochastic hierarchical blockmodel with $r$-complex contagion, under certain mild technical assumptions, the optimal seeding strategy is to put all the seeds in a single community, if, for each vertex-pair $(u,v)$, the probability that the edge $(u,v)$ is included satisfies $p_{uv}=\omega(n^{-(1+1/r)})$. Notice that the assumption $p_{uv}=\omega(n^{-(1+1/r)})$ captures many real life social networks. In fact, it is well-known that an Erd\H{o}s-R\'{e}nyi graph $\calG(n,p)$ with $p=o(1/n)$ is globally disconnected: with probability $1-o(1)$, the graph consists of a union of tiny connected components, each of which has size $O(\log n)$.

The technical heart of this result is a novel coupling argument in Proposition~\ref{prop:main}.   We simultaneously couple four cascade processes to compare two probabilities: 1) the probability of infection spreading throughout an Erd\H{o}s-R\'{e}nyi graph after the $(k + 1)$-st seed, conditioned on  not already being entirely infected after $k$ seeds; 2) the probability of infection spreading throughout the same graph after the $(k + 2)$-nd seed, conditioned on not already being entirely infected after $k+1$ seeds.  This shows that the marginal rate of infection always goes up, revealing the ``supermodular'' nature of the $r$-complex contagion.
The supermodular property revealed by Proposition~\ref{prop:main} is a property for cascade behavior on Erd\H{o}s-R\'{e}nyi random graphs in general, so it is also interesting on its own.

Our result is in sharp contrast to Balkanski et al.'s observation.  Balkanski et al.~\cite{balkanski2017importance} studies the stochastic blockmodel with a well-studied submodular cascade model, \emph{the independent cascade model}, and remarks that ``when an influential node from a certain community is selected to initiate a cascade, the marginal contribution of adding another node from that same community is small, since the nodes in that community were likely already influenced.''

\paragraph*{Algorithmic aspects}

For influence maximization in submodular cascades, a greedy algorithm efficiently finds a seeding set with influence at least a $(1-1/e)$ fraction of the optimal~\cite{kempe2003maximizing}, and much of the work following Kempe et al.~\cite{kempe2003maximizing}, which proposed the greedy algorithm, has attempted to make greedy approaches efficient and scalable~\cite{chen2009efficient, chen2010scalable, lucier2015influence, cohen2014sketch, tullock1980toward, schoenebeck2019influence}.

Greedy approaches, unfortunately, can perform poorly in the nonsubmodular setting~\cite{angell2017don}.  Moreover, in contrast to the submodular case which has efficient constant approximation algorithms, for general nonsubmodular cascades, it is \textsf{NP}-hard even to approximate influence maximization to within an $\Omega(n^{1-\epsilon})$ factor of the optimal~\cite{kempe2005influential}.  This inapproximability result has been extended to several much more restrictive nonsubmodular models~\cite{chen2016robust,li2017influence,schoenebeck2017beyond,schoenebeck2019beyond}.
Intuitively, nonsubmodular influence maximization is hard because the potential synergy of multiple seeds makes it necessary to consider groups of seeds rather than just individual seeds.
In contrast, with submodular influence maximization, not much is lost by considering seeds one at a time in a myopic way.

Can the $\Omega(n^{1-\epsilon})$ inapproximability results of~Kempe et al.~\cite{kempe2005influential} be circumvented if we further assume the stochastic hierarchical blockmodel?
On the one hand, the stochastic hierarchical structure seems optimized for a dynamic programming approach: perform dynamic programming from the bottom to the root in the tree-like community structure.  On the other hand,  Schoenebeck and Tao~\cite{schoenebeck2017beyond,schoenebeck2019beyond} show that the $\Omega(n^{1-\epsilon})$ inapproximability results extend to the setting where the networks are stochastic hierarchical blockmodels.

\noindent\textbf{Result 2:} However, Result 1 (when the network is reasonably dense,  putting all the seeds in a single community is optimal)  can naturally be extended to a dynamic programming algorithm.  We show that this algorithm is optimal if the probability $p_{uv}$ that each edge appears does not fall into a narrow regime.
Interestingly, a heuristic based on dynamic programming works fairly well in practice~\cite{angell2017don}.
Our second result theoretically justifies the success of this approach, at least in the setting of $r$-complex contagions.

\section{Preliminaries}
We study complex contagions on social networks with community structures.  This section defines the complex contagion and our model for social networks with community structures.

\subsection{\texorpdfstring{$r$}{r}-Complex Contagion}
Given a social network modeled as an undirected graph $G=(V,E)$, in a cascade, a subset of nodes $S\subseteq V$ is chosen as the seed set; these seeds, being infected, then spread their influence across the graph according to some specified model.

In this paper, we consider a well-known cascade model named \emph{$r$-complex contagion}, also known as \emph{bootstrap percolation} and the \emph{fixed threshold model}: a node is infected if and only if at least $r$ of its neighbors are infected.  We use $\sigma_{r,G}(S)$ to denote the total number of infected vertices at the end of the cascade, and $\sigma_{r,\calG}(S) =\E_{G\sim\calG}\left[\sigma_{r,G}(S)\right]$ if the graph $G$ is sampled from some distribution $\calG$.  Notice that the function $\sigma_{r,G}(\cdot)$ is deterministic once the graph $G$ and $r$ are fixed.

\paragraph{Submodularity of a cascade model}
Other than the $r$-complex contagion, most cascade models are stochastic: the total number of infected vertices is not deterministic but rather a \emph{random variable}.
$\sigma_G(S)$ usually refers to the \emph{expected} number of infected vertices given the seed set $S$.
A cascade model is \emph{submodular} if, given any graph and $S\subseteq T\subseteq V$ and any vertex $v\in V\setminus T$, we have
$$\sigma_G(S\cup\{v\})-\sigma_G(S)\geq\sigma_G(T\cup\{v\})-\sigma_G(T),$$
and it is \emph{nonsubmodular} otherwise.
Typical submodular cascade models include \emph{the linear threshold model} and \emph{the independent cascade model}~\cite{kempe2003maximizing}, which are studied in an enormous past literature.
The $r$-complex contagion, on the other hand, is a paradigmatic nonsubmodular model.

\subsection{Stochastic Hierarchical Blockmodels}\label{sec:SHB}
We study the \emph{stochastic hierarchical blockmodel} first introduced in~\cite{schoenebeck2019beyond}.
The stochastic hierarchical blockmodel is a special case of the \emph{stochastic blockmodel}~\cite{holland1983stochastic}.
Intuitively, the stochastic blockmodel is a stochastic graph model generating networks with community structures, and the stochastic hierarchical blockmodel further assumes that the communities form a hierarchical structure.
Our definition in this section follows closely to~\cite{schoenebeck2019beyond}.

\begin{definition}\label{defi:SHB}
  A \emph{stochastic hierarchical blockmodel} is a distribution $\calG=(V,T)$ of unweighted undirected graphs sharing the same vertex set $V$, and $T=(V_T,E_T,w)$ is a weighted tree $T$ called a \emph{hierarchy tree}.
  The third parameter is the weight function $w:V_T\to[0,1]$ satisfying $w(t_1)<w(t_2)$ for any $t_1,t_2\in V_T$ such that $t_1$ is an ancestor of $t_2$.  Let $L_T\subseteq V_T$ be the set of leaves in $T$.  Each leaf node $t\in L_T$ corresponds to a subset of vertices $V(t)\subseteq V$, and the $V(t)$ sets partition the vertices in $V$.
  In general, if $t\not\in L_T$, we denote $V(t)=\bigcup_{t'\in L_T:t'\text{ is an offspring of }t}V(t')$.

  The graph $G=(V,E)$ is sampled from $\calG$ in the following way.
  The vertex set $V$ is deterministic.
  For $u, v \in V$, the edge $(u, v)$ appears in $G$ with probability equal to the weight of the least common ancestor of $u$ and $v$ in $T$.
  That is $\Pr((u,v)\in E) = \max_{t: u, v \in V(t)} w(t)$.
\end{definition}

In the rest of this paper, we use the words ``tree node'' and ``vertex'' to refer to the vertices in $V_T$ and $V$ respectively.
In Definition~\ref{defi:SHB}, the tree node $t\in V_T$ corresponds to community $V(t)\subseteq V$  in the social network.  Moreover, if $t$ is not a leaf and $t_1,t_2,\ldots$ are the children of $t$ in  $V_T$, then $V(t_1), V(t_2),\ldots$ partition $V(t)$ into sub-communities.
Thus, our assumption that for any $t_1,t_2\in V_T$ where $t_1$ is an ancestor of $t_2$ we have $w(t_1)<w(t_2)$ implies that the relation between two vertices is stronger if they are in a same sub-community in a lower level, which is natural.

To capture the scenario where the advertiser has the information on the high-level community structure but lacks the knowledge of the detailed connections inside the communities, when defining the influence maximization problem as an optimization problem, we would like to include $T$ as a part of input, but not $G$.
Rather than choosing which specific vertices are seeds, the seed-picker decides the number of seeds on each leaf and the graph $G\sim\calG(n,T)$ is realized after seeds are chosen.
Moreover, we are interested in large social networks with $n\rightarrow\infty$, so we would like that a single encoding of $T$ is compatible with varying $n$.
To enable this feature, we consider the following variant of the stochastic hierarchical block model.

\begin{definition}\label{defi:SHBL}
  A \emph{succinct stochastic hierarchical blockmodel} is a distribution $\calG(n,T)$ of unweighted undirected graphs sharing the same vertex set $V$ with $|V|=n$, where $n$ is an integer which is assumed to be extremely large.  The hierarchy tree $T=(V_T,E_T,w,v)$ is the same as it is in Definition~\ref{defi:SHB}, except for the followings.
  \begin{enumerate}
      \item Instead of mapping a tree node $t$ to a weight in $[0,1]$, the weight function $w:V_T\to\mathcal{F}$ maps each tree node to a function $f \in \mathcal{F}=\{f\mid f:\mathbb{Z}^+\to[0,1]\}$ which maps an integer (denoting the number of vertices in the network) to a weight in $[0,1]$. The weight of $t$ is then defined by $(w(t))(n)$. We assume $\mathcal{F}$ is the space of all functions that can be succinctly encoded.
      \item The fourth parameter $v:V_T\to(0,1]$ maps each tree node $t \in V_T$ to the fraction of vertices in $V(t)$.  That is: $v(t)=|V(t)|/n$. Naturally, we have $\sum_{t\in L_T}v(t)=1$ and $\sum_{t':t'\text{ is a child of }t}v(t')=v(t)$.
  \end{enumerate}
\end{definition}

We assume throughout that $\calG(n,T)$ has the following properties.
\begin{description}
\item[Large communities] For tree node $t\in V_T$, because $v(t)$ does not depend on $n$, $|V(t)| = v(t)n = \Theta(n)$. In particular, $|V(t)|$ goes to infinity as $n$ does.
\item[Proper separation] $w(t_1) = o\left(w(t_2)\right)$ for any $t_1,t_2\in V_T$ such that $t_1$ is an ancestor of $t_2$.  That is, the connection between sub-community $t_2$ is asymptotically (with respect to $n$) denser than its super-community $t_1$.
\end{description}
Our definitions of $w$ and $v$ are designed so that we can fix a hierarchy tree $T=(V_T,E_T,w,v)$ and naturally define  $\calG(n,T)$ for any $n$.
As we will see in the next subsection, this allows us to take $T$ as input and then allow $n\rightarrow\infty$ when considering \infmax (to be defined soon).
This enables us to consider graphs having exponentially many vertices.

Finally, we define the \emph{density} of a tree node.

\begin{definition}\label{defi:density}
Given a hierarchy tree $T=(V_T,E_T,w,v)$ and a tree node $t\in V_T$, the \emph{density} of the tree node is $\rho(t)=w(t)\cdot(v(t)n)^{1/r}.$
\end{definition}

\subsection{The \infmax Problem}
We study the $r$-complex contagion on the succinct stochastic hierarchical blockmodel.  Roughly speaking, given hierarchy tree $T$ and an integer $K$, we want to choose $K$ seeds which maximize the expected total number of infected vertices, where the expectation is taken over the graph sampling $G\sim\calG(n,T)$ as $n\rightarrow\infty$.

\begin{definition}\label{defi:infmax}
  The \emph{influence maximization problem} \infmax is an optimization problem which takes as inputs an integer $r$, a hierarchy tree $T=(V_T,E_T,w,v)$ as in Definition~\ref{defi:SHBL}, and an integer $K$, and outputs $\bm{k}\in\N_{\geq0}^{|L_T|}$---an allocation of $K$ seeds into the leaves $L_T$ with $\sum_{t\in L_T} k_t=K$ that maximizes
  $$\Sigma_{r,T}(\bm{k}):=\lim_{n\rightarrow\infty}\frac{\E_{G\sim\calG(n,T)}\left[\sigma_{r,G}(S_{\bm{k}})\right]}n,\footnote{We divide the expected number of infected vertices by $n$ to avoid an infinite limit.  However, as a result, our analysis naturally ignores lower order terms.}$$
  the expected fraction of infected vertices in $\calG(n,T)$ with the seeding strategy defined by $\bm{k}$, where $S_{\bm{k}}$ denotes the seed set in $G$ generated according to $\bm{k}$.
\end{definition}

Before we move on, the following remark is very important throughout the paper.

\begin{remark}\label{remark:asy}
In Definition~\ref{defi:infmax}, $n$ is not part of the inputs to the \infmax instance.
Instead, the tree $T$ is given as an input to the instance, and we take $n\rightarrow\infty$ to compute $\Sigma_{r,T}(\bm{k})$ \emph{after} the seed allocation is determined.
Therefore, asymptotically, all the input parameters of the instance, including $K,r$ and the encoding size of $T$, are \emph{constants} with respect to $n$.
Thus, there are two different asymptotic scopes in this paper: \emph{the asymptotic scope with respect to the input size} and \emph{the asymptotic scope with respect to $n$}.
Naturally, when we are analyzing the running time of an \infmax algorithm, we should use the asymptotic scope with respect to the input size, not of $n$.
On the other hand, when we are analyzing the number of infected vertices after the cascade, we should use the asymptotic scope with respect to $n$.

In this paper, we use $O_I(\cdot),\Omega_I(\cdot),\Theta_I(\cdot),o_I(\cdot),\omega_I(\cdot)$ to refer to the asymptotic scope with respect to the input size, and we use $O(\cdot),\Omega(\cdot),\Theta(\cdot),o(\cdot),\omega(\cdot)$ to refer to the asymptotic scope with respect to $n$.
For example, with respect to $n$ we always have $r=\Theta(1)$, $K=\Theta(1)$ and $|V_T|=\Theta(1)$.
\end{remark}

Lastly, we have assumed that $r \geq 2$, so that the contagion is nonsubmodular.
When $r=1$, the cascade model becomes a special case of the \emph{independent cascade model}~\cite{kempe2003maximizing}, which is a submodular cascade model.
As mentioned, for submodular  \infmax, a simple greedy algorithm is known to achieve a $(1-1/e)$-approximation to the optimal influence~\cite{kempe2003maximizing,kempe2005influential,MosselR10}.

\subsection{Complex Contagion on Erd\H{o}s-R\'{e}nyi Graphs}\label{sect:preliminary}
In this section, we consider the $r$-complex contagion on the Erd\H{o}s-R\'{e}nyi random graph $\calG(n,p)$.
We review some results from~\cite{janson2012bootstrap} which are used in our paper.

\begin{definition}
  The \emph{Erd\H{o}s-R\'{e}nyi random graph} $\calG(n,p)$ is a distribution of graphs with the same vertex set $V$ with $|V|=n$ and we include an edge $(u,v)\in E$ with probability $p$ independently for each pair of vertices $u,v$.
\end{definition}

The \infmax problem in Definition~\ref{defi:infmax} on $\calG(n,p)$ is trivial, as there is only one possible allocation of the  $K$ seeds: allocate all the seeds to the single leaf node of $T$, which is the root.
Therefore, $\sigma_{r,T}(\cdot)$ in Definition~\ref{defi:infmax} depends only on the \emph{number} of seeds $K=|\bm{k}|$, not on the seed allocation $\bm{k}$ itself.
In this section, we slightly abuse the notation $\sigma$ such that it is a function mapping an \emph{integer} to $\R_{\geq0}$ (rather than mapping \emph{an allocation of $K$ seeds} to $\R_{\geq0}$ as it is in Definition~\ref{defi:infmax}), and let $\sigma_{r,\calG(n,p)}(k)$ be the expected number of infected vertices after the cascade given $k$ seeds.
Correspondingly, let $\sigma_{r,G}(k)$ be the actual number of infected vertices after the graph $G$ is sampled from $\calG(n,p)$.

\begin{theorem}[A special case of Theorem~3.1 in~\cite{janson2012bootstrap}]\label{thm:gnp_sub}
  Suppose $r\geq2$, $p=o(n^{-1/r})$ and $p = \omega(n^{-1})$. We have
  \begin{enumerate}
  \item if $k$ is a constant, then $\sigma_{r,\calG(n,p)}(k)\leq 2k$ with probability $1-o(1)$;
  \item if $k = \omega\left((1/np^r)^{1/(r-1)}\right)$, then $\sigma_{r,\calG(n,p)}(k)=n-o(n)$ with probability $1-o(1)$.
  \end{enumerate}
\end{theorem}

\begin{theorem}[Theorem~5.8 in~\cite{janson2012bootstrap}]\label{thm:gnp_sup}
  If $r\geq2$, $p=\omega(n^{-1/r})$ and $k\geq r$, then $\Pr_{G\sim\calG(n,p)}\left[\sigma_{r,G}(k)=n\right]=1-o(1)$.
\end{theorem}

When $p=\Theta(n^{-1/r})$, the probability that $k$ seeds infect all the $n$ vertices is positive, but bounded away from $1$.
We use $\po(\lambda)$ to denote the Poisson distribution with mean $\lambda$.

\begin{theorem}[Theorem~5.6 and Remark~5.7 in~\cite{janson2012bootstrap}]\label{thm:gnp_boundary}
  If $r\geq2$, $p=cn^{-1/r}+o(n^{-1/r})$ for some constant $c>0$, and $k\geq r$ is a constant, then
  $$\lim_{n\rightarrow\infty}\Pr\left(\sigma_{r,\calG(n,p)}(k)=n\right)=\zeta(k,c),$$
  for some $\zeta(k,c)\in(0,1)$. Furthermore, there exist numbers $\zeta(k,c,\ell)>0$ for $\ell\geq k$ such that
  $$\lim_{n\rightarrow\infty}\Pr\left(\sigma_{r,\calG(n,p)}(k)=\ell\right)=\zeta(k,c,\ell)$$
  for each $\ell\geq k$, and $\zeta(k,c)+\sum_{\ell=k}^\infty\zeta(k,c,\ell)=1$.

  Moreover, the numbers $\zeta(k,c,\ell)$'s and $\zeta(k,c)$ can be expressed as the hitting probabilities of the following inhomogeneous random walk.
  Let $\xi_\ell\sim\po\left(\binom{\ell-1}{r-1}c^r\right)$, $\ell\geq1$ be independent, and let $\tilde{S}_\ell:=\sum_{j=1}^\ell(\xi_j-1)$ and $\tilde{T}:=\min\{\ell:k+\tilde{S}_\ell=0\}\in\N\cup\{\infty\}$. Then
  \begin{equation}\label{eq:zeta}
      \zeta(k,c)=\Pr\left(\tilde{T}=\infty\right)=\Pr\left(k+\tilde{S}_\ell\geq1\mbox{ for all }\ell\geq1\right)
  \end{equation}
  and $\zeta(k,c,\ell)=\Pr(\tilde{T}=\ell)$.
\end{theorem}

We have the following corollary for Theorem~\ref{thm:gnp_boundary}, saying that when $p=\Theta(n^{-1/r})$, if not all vertices are infected, then the number of infected vertices is constant.
As a consequence, if the cascade spreads to more than constantly many vertices, then all vertices will be infected.

\begin{corollary}[Lemma~11.4 in~\cite{janson2012bootstrap}]\label{cor:gnp_boundary}
If $r\geq2$, $p=cn^{-1/r}+o(n^{-1/r})$ for some constant $c>0$, and $k\geq r$, then
$$\lim_{n\rightarrow\infty}\Pr\left(\phi(n)\leq\sigma_{r,\calG(n,p)}(k)<n\right)=0$$
for any function $\phi:\mathbb{Z}^+\to\mathbb{R}^+$ such that $\lim_{n\rightarrow\infty}\phi(n)=\infty$.
\end{corollary}

\section{Our main result}
Our main result is the following theorem, which states that the optimal seeding strategy is to put all the seeds in a community with the highest density, when the root has a weight in $\omega(1/n^{1+1/r})$.

\begin{theorem}
\label{THM:MAIN}
Consider the \infmax problem with $r\geq2$, $T=(V_T,E_T,w,v)$, $K>0$ and the weight of the root node satisfying $w(\text{root})=\omega(1/n^{1+1/r})$.
Let $\displaystyle t^\ast\in\argmax_{t\in L_T}\rho(t)$ and $\bm{k}^\ast$ be the seeding strategy that puts all the $K$ seeds on $t^\ast$. Then $\displaystyle \bm{k}^\ast\in\argmax_{\bm{k}}\Sigma_{r,T}(\bm{k})$.
\end{theorem}

Notice that the assumption $w(\text{root})=\omega(1/n^{1+1/r})$ captures many real life social networks.
In fact, it is well-known that an Erd\H{o}s-R\'{e}nyi graph $\calG(n,p)$ with $p=o(1/n)$ is globally disconnected: with probability $1-o(1)$, the graph consists of a union of tiny connected components, each of which has size $O(\log n)$.

The remaining part of this section is dedicated to proving Theorem~\ref{THM:MAIN}.
We assume $w(\text{root})=\omega(1/n^{1+1/r})$ in this section from now on.
It is worth noting that, in many parts of this proof, and also in the proof of Theorem~\ref{thm:isolated}, we have used the fact that an infection of $o(n)$ vertices contributes $0$ to the objective $\Sigma_{r,T}(\bm{k})$, as we have taken the limit $n\rightarrow\infty$ and divided the expected number of infections by $n$ in Definition~\ref{defi:infmax}.

\begin{definition}
Given $T=(V_T,E_T,w,v)$, a tree node $t\in V_T$ is \emph{supercritical} if $w(t)=\omega(1/n^{1/r})$, is \emph{critical} if $w(t)=\Theta(1/n^{1/r})$, and is \emph{subcritical} if $w(t)=o(1/n^{1/r})$.
\end{definition}

From the results in Sect.~\ref{sect:preliminary}, if we allocate $k\geq r$ seeds on a supercritical leaf $t\in L_T$, then with probability $1-o(1)$ all vertices in $V(t)$ will be infected; if we allocate $k$ seeds on a subcritical leaf $t\in L_T$, at most a constant number of vertices, $2k=\Theta(1)$, will be infected; if we allocate $k\geq r$ seeds on a critical leaf $t\in L_T$, the number of infected vertices in $V(t)$ follows Theorem~\ref{thm:gnp_boundary}.

We say a tree node $t\in V_T$ is \emph{activated} in a cascade process if the number of infected vertices in $V(t)$ is $v(t)n-o(n)$, i.e., almost all vertices in $V(t)$ are infected.
Given a seeding strategy $\bm{k}$, let $P_{\bm{k}}$ be the probability that at least one tree node is activated when $n\rightarrow\infty$.
Notice that this is equivalent to at least one leaf being activated.
The proof of Theorem~\ref{THM:MAIN} consists of two parts.
We will first show that, $P_{\bm{k}}$ completely determines $\Sigma_{r,T}(\bm{k})$ (Lemma~\ref{lemmaA}).
Secondly, we show that placing all the seeds on a single leaf with the maximum density will maximize $P_{\bm{k}}$  (Lemma~\ref{lemmaB}).

\begin{lemma}
\label{lemmaA}
Given any two seeding strategies $\bm{k}_1,\bm{k}_2$, if $P_{\bm{k}_1}\leq P_{\bm{k}_2}$, then $\Sigma_{r,T}(\bm{k}_1)\leq\Sigma_{r,T}(\bm{k}_2)$.
\end{lemma}

\begin{lemma}\label{lemmaB}
Let $\bm{k}$ be the seeding strategy that allocates all the $K$ seeds on a leaf $\displaystyle t^\ast\in\argmax_{t\in L_T}(\rho(t))$. Then $\bm{k}$ maximizes $P_{\bm{k}}$.
\end{lemma}

Lemma~\ref{lemmaA} and Lemma~\ref{lemmaB} imply Theorem~\ref{THM:MAIN}.

\subsection{Proof Sketch of Lemma~\ref{lemmaA}}
We sketch the proof.  The full proof is in the appendix.
\begin{proof}[Proof (sketch)]
Let $E$ be the event that at least one leaf (or tree node) is activated at the end of the cascade.

In the case that $E$ does not happen, we show there are only $o(n)$ infected vertices in $V$, regardless of the seeding strategy.
First, Theorem~\ref{thm:gnp_sup} and Corollary~\ref{cor:gnp_boundary} imply that the number of infected vertices in a critical or supercritical leaf $t$, with high probability, can only be either a constant or $v(t)n$.  Because $E$ does not happen, it must be the former with high probability.
Second, Theorem~\ref{thm:gnp_sub} indicates that a subcritical leaf with a constant number of seeds will not have $\omega(1)$ infected vertices with high probability.
As there are only a constant number of infections in each of the critical or supercritical leaves, and we have only a constant number $K=\Theta(1)$ of seeds, this implies that there are also only a constant number of infections in subcritical leaves.

If $E$ happens, we can show that the expected total number of infected vertices does not vary significantly for different seeding strategies.
Consider two leaves $t_1,t_2$ with their least common ancestor $t$.
If the leaf $t_1$ is activated, we find a lower bound of the probability that a vertex $v\in V(t_2)$ is infected due to the influence of $V(t_1)$.
We assume without loss of generality that $w(t)=o(1/n)$, which can only further reduce $v$'s infection probability from the case when $w(t)$ is in $\Omega(1/n)$.
With this assumption, the probability that $v\in V(t_2)$ is infected by the vertices in $V(t_1)$ is
$$\binom{v(t_1)n}{r}w(t)^r(1-w(t))^{v(t_1)n-r}=\omega\left(n^r\left(\frac1{n^{1+\frac1r}}\right)^r\cdot1\right)=\omega\left(\frac1n\right),$$
where the first equality uses the assumption $w(t)=o(1/n)$ so that $(1-w(t))^{v(t_1)n-r}=\Omega(1)$.
Thus, with high probability, there are $\omega(1/n)\cdot\Theta(n)=\omega(1)$ infected vertices in $V(t_2)$.
Theorem~\ref{thm:gnp_sup} and Corollary~\ref{cor:gnp_boundary} show that $t_2$ will be activated with high probability if $t_2$ is critical or supercritical.
Therefore, when $E$ happens, all the critical and supercritical leaves will be activated.
As for subcritical leaves, the number of infected vertices may vary, but Theorem~\ref{thm:gnp_sub} intuitively suggests that adding a constant number of seeds is insignificant (we handle this rigorously in the full proof).
Therefore, the expected total number of infections equals to the number of vertices in all critical and supercritical leaves, plus the expected number of infected vertices in subcritical leaves which does not significantly depend on the seeding strategy $\bm{k}$.

In conclusion, the number of infected vertices only significantly depends on whether or not $E$ happens.
In particular, we have a fixed fraction of infected vertices whose size does not depend on $\bm{k}$ if $E$ happens, and a negligible number of infected vertices if $E$ does not happen.
Therefore, $P_{\bm{k}}$ characterizes $\Sigma_{r,T}(\bm{k})$, and a larger $P_{\bm{k}}$ implies a larger $\Sigma_{r,T}(\bm{k})$.
\end{proof}

\subsection{Proof of Lemma~\ref{lemmaB}}
We first handle some corner cases.
If $K<r$, then the cascade will not even start, and any seeding strategy is considered optimal.
If $T$ contains a supercritical leaf, the leaf with the highest density is also supercritical.
Putting all the $K\geq r$ seeds in this leaf, by Theorem~\ref{thm:gnp_sup}, will activate the leaf with probability $1-o(1)$.
Therefore, this strategy makes $P_{\bm{k}}=1$, which is clearly optimal.
In the remaining part of this subsection, we shall only consider $K\geq r$ and all the leaves are either critical or subcritical.
Notice that, by the proper separation assumption, all internal tree nodes of $T$ are subcritical.

We split the cascade process into two phases.
In Phase I, we restrict the cascade within the leaf blocks ($V(t)$ where $t\in L_T$), and temporarily assume there are no edges between two different leaf blocks (similar to if $w(t) = 0$ for all $t \not\in L_T$).  After Phase I, Phase II consists of the remaining cascade process.

Proposition~\ref{prop:isolate} shows that maximizing $P_{\bm{k}}$ is equivalent to maximizing the probability that a leaf is activated in Phase I.
Therefore, we can treat $T$ such that all the leaves, each of which corresponds to a $\calG(n,p)$ random graph, are isolated.

\begin{proposition}\label{prop:isolate}
If no leaf is activated after Phase I, then with probability $1-o(1)$ no vertex will be infected in Phase II, i.e., the cascade will end after Phase I.
\end{proposition}
We sketch the proof here, and the full proof is available in the appendix.
\begin{proof}[Proof (sketch)]
Consider any critical leaf $t$ and an arbitrary vertex $v\in V(t)$ that is not infected after Phase I.
Let $K_{in}$ be the number of infected vertices in $V(t)$ after Phase I, and $K_{out}$ be the number of infected vertices in $V\setminus V(t)$.
If no leaf is activated after Phase I, Theorem~\ref{thm:gnp_sub} and Corollary~\ref{cor:gnp_boundary} show that $K_{in}=O(1)$ and $K_{out}=O(1)$ with high probability.
The probability that $v$ is connected to any of the $K_{in}$ infected vertices in $V(t)$ can only be less than $w(t)=\Theta(n^{-1/r})$ conditioning on that the cascade inside $V(t)$ does not carry to $v$, so the probability that $v$ has $a<r$  infected neighbors in $V(t)$ is $O(n^{-a/r})$.
On the other hand, the probability that $v$ has $r-a$ neighbors among the $K_{out}$ outside infected vertices is $o(n^{-(r-a)/r})$. Therefore, the probability that $v$ is infected in the next iteration is $\sum_{a=0}^{r-1}O(n^{-a/r})\cdot o(n^{-(r-a)/r})=o(1/n)$, and the expected total number of vertices infected in the next iteration after Phase I is $o(1)$.
The proposition follows from the Markov's inequality.
\end{proof}

Since Theorem~\ref{thm:gnp_sub} shows that any constant number of seeds will not activate a subcritical leaf with high probability, we should only consider putting seeds in critical leaves.
In Proposition~\ref{prop:main}, we show that in a critical leaf $t$, the probability that the $(i+1)$-th seed will activate $t$ conditioning on the first $i$ seeds failing to do so is increasing as $i$ increases.
Intuitively, Proposition~\ref{prop:main} reveals a super-modular nature of the $r$-complex contagion on a critical leaf, making it beneficial to put all seeds together so that the synergy effect is maximized, which intuitively implies Lemma~\ref{lemmaB}.
The proof of Proposition~\ref{prop:main} is the most technical result of this paper, we will present it in Sect.~\ref{sec:coupling}.

\begin{proposition}[log-concavity of $\lim\limits_{n\rightarrow\infty}\Pr(E_k^n)$]\label{prop:main}
Consider an Erd\H{o}s-R\'{e}nyi random graph $\calG(n,p)$ with $p=cn^{-1/r}+o(n^{-1/r})$, and assume an arbitrary order on the $n$ vertices.
Let $E_k^n$ be the event that seeding the first $k$ vertices does not make all the $n$ vertices infected.
We have
$\lim\limits_{n\rightarrow\infty}\Pr(E_{k+2}^n\mid E_{k+1}^n)<\lim\limits_{n\rightarrow\infty}\Pr(E_{k+1}^n\mid E_k^n)$ for any $k\geq r-1$.
\end{proposition}

Equipped with Proposition~\ref{prop:main}, to show Lemma~\ref{lemmaB}, we show that the seeding strategy that allocates $K_1>0$ seeds on a critical leaf $t_1$ and $K_2>0$ seeds on a critical leaf $t_2$ cannot be optimal.
Firstly, it is obvious that both $K_1$ and $K_2$ should be at least $r$, for otherwise those $K_1$ ($K_2$) seeds on $t_1$ ($t_2$) are simply wasted.

Let $E_k^n$ be the event that the first $k$ seeds on $t_1$ fail to activate $t_1$ and $F_k^n$ be the event that the first $k$ seeds on $t_2$ fail to activate $t_2$.
By Proposition~\ref{prop:main}, we have
$\lim\limits_{n\rightarrow\infty}\Pr(E_{K_1+1}^n\mid E_{K_1}^n)<\lim\limits_{n\rightarrow\infty}\Pr(E_{K_1}^n\mid E_{K_1-1}^n)$
and
$\lim\limits_{n\rightarrow\infty}\Pr(F_{K_2+1}^n\mid F_{K_2}^n)<\lim\limits_{n\rightarrow\infty}\Pr(F_{K_2}^n\mid F_{K_2-1}^n)$,
which implies
\begin{align*}
    &\lim\limits_{n\rightarrow\infty}\frac{\Pr(E_{K_1+1}^n)\Pr(F_{K_2-1}^n)}{\Pr(E_{K_1}^n)\Pr(F_{K_2}^n)}\cdot\frac{\Pr(E_{K_1-1}^n)\Pr(F_{K_2+1}^n)}{\Pr(E_{K_1}^n)\Pr(F_{K_2}^n)}\\
    =&\lim\limits_{n\rightarrow\infty}\frac{\Pr(E_{K_1+1}^n\mid E_{K_1}^n)\Pr(F_{K_2+1}^n\mid F_{K_2}^n)}{\Pr(E_{K_1}^n\mid E_{K_1-1}^n)\Pr(F_{K_2}^n\mid F_{K_2-1}^n)}<1.
\end{align*}
Therefore, we have either $\lim\limits_{n\rightarrow\infty}\frac{\Pr(E_{K_1+1}^n)\Pr(F_{K_2-1}^n)}{\Pr(E_{K_1}^n)\Pr(F_{K_2}^n)}$ or $\lim\limits_{n\rightarrow\infty}\frac{\Pr(E_{K_1-1}^n)\Pr(F_{K_2+1}^n)}{\Pr(E_{K_1}^n)\Pr(F_{K_2}^n)}$ is less than $1$.
This means either the strategy putting $K_1+1$ seeds on $t_1$ and $K_2-1$ seeds on $t_2$, or the strategy putting $K_1-1$ seeds on $t_1$ and $K_2+1$ seeds on $t_2$ makes it more likely that at least one of $t_1$ and $t_2$ is activated.
Therefore, the strategy putting $K_1$ and $K_2$ seeds on $t_1$ and $t_2$ respectively cannot be optimal.  This implies an optimal strategy should not allocate seeds on more than one leaf.

Finally, a critical leaf $t$ with $v(t)n$ vertices and weight $w(t)$ can be viewed as an Erd\H{o}s-R\'{e}nyi random graph $\calG(m,p)$ with $m=v(t)n$ and
$p=w(t)=\rho(t)\cdot (v(t)n)^{-1/r}=\rho(t)m^{-1/r}$, where $\rho(t)=\Theta(1)$ when $t$ is critical.
Taking $c=\rho(t)$ in Theorem~\ref{thm:gnp_boundary}, we can see that $\xi_\ell$ has a larger Poisson mean if $c$ is larger, making it more likely that the $\calG(m,p)$ is fully infected (to see this more naturally, larger $c$ means larger $p$ if we fix $m$).
Thus, given that we should put all the $K$ seeds in a single leaf, we should put them on a leaf with the highest density.
This concludes Lemma~\ref{lemmaB}.

\section{Proof for Proposition~\ref{prop:main}}\label{sec:coupling}

Since the event $E_{k+1}^n$ implies $E_k^n$, we have
$\Pr(E_{k+1}^n|E_k^n)=\Pr(E_{k+1}^n)/\Pr(E_k^n)$.  Therefore, the inequality we are proving is equivalent to
$
\lim\limits_{n\rightarrow\infty}\Pr(E_{k+2}^n)/\Pr(E_{k+1}^n)<\lim\limits_{n\rightarrow\infty}\Pr(E_{k+1}^n)/\Pr(E_{k}^n)$,
and it suffices to show that
\begin{equation}\label{eqn:prop1}
\lim\limits_{n\rightarrow\infty}\Pr(E_{k+2}^n)\lim\limits_{n\rightarrow\infty}\Pr(E_k^n)<\lim\limits_{n\rightarrow\infty}\Pr(E_{k+1}^n)\lim\limits_{n\rightarrow\infty}\Pr(E_{k+1}^n).
\end{equation}
Proposition~\ref{prop:main} shows that the failure probability, $\lim\limits_{n\rightarrow\infty}\Pr(E_k^n)$, is logarithmically concave with respect to $k$.

The remaining part of the proof is split into four parts:  In Sect.~\ref{sect:prop11}, we begin by translating Eqn.~\eqref{eqn:prop1} in the language of inhomogeneous random walks.  In Sect.~\ref{sect:prop12}, we present a coupling of two inhomogeneous random walks to prove Eqn.~\eqref{eqn:prop1}.   In Sect.~\ref{sect:prop13}, we prove the validity of the coupling. in Sect.~\ref{sect:prop14}, we finally show the coupling implies Eqn.~\eqref{eqn:prop1} .

\subsection{Inhomogeneous random walk interpretation}\label{sect:prop11}
We adopt the inhomogeneous random walk interpretation from Theorem~\ref{thm:gnp_boundary}, and view the event $E_k^n$ as the following:
The random walk starts at $x=k$; in the $i$-th iteration, $x$ moves to the left by $1$ unit, and moves to the right by $\alpha(i)\sim\po\left(\binom{i-1}{r-1}c^r\right)$ units; let $\mathcal{E}_k$ be the event that the random walk reaches $x = 0$.  By Theorem~\ref{thm:gnp_boundary}, $\Pr(\mathcal{E}_k) = \lim\limits_{n\rightarrow\infty}\Pr(E_k^n)$.  Thus,  $\lim\limits_{n\rightarrow\infty}\Pr(E_{k+2}^n)\lim\limits_{n\rightarrow\infty}\Pr(E_k^n) = \Pr(\mathcal{E}_{k+2})\Pr(\mathcal{E}_{k})$.  In this proof, we let $\lambda(i)=\binom{i-1}{r-1}c^r$, and in particular, $\lambda(0)=\lambda(1)=\cdots=\lambda(r-1)=0$.  Note that as $i$ increases, the expected movement of the walk increases, and make it harder to reach $0$.  This observation is important for our proof.

To compute $\Pr(\mathcal{E}_{k+2})\Pr(\mathcal{E}_{k})$, we consider the following process. A random walk in $\Z^2$ starts at $(k+2, k)$.  In each iteration $i$, the random walk moves from $(x,y)$ to $(x-1+\alpha(i),y-1+\beta(i))$ where $\alpha(i)$ and $\beta(i)$ are sampled from $\po(\lambda(i))$ independently.
If the random walk hits the axis $y = 0$ after a certain iteration $\mathcal{T}$, then it is stuck to the axis, i.e., for any $i>\mathcal{T}$, the update in the $i$-th iteration is from $(x,0)$ to $(x-1+\alpha(i),0)$; similarly, after reaching the axis $x = 0$, the random walk is stuck to the axis $x = 0$ and updates to $(0,y-1+\beta(i))$.
Then, $\Pr(\mathcal{E}_{k+2})\Pr(\mathcal{E}_{k})$ is the probability that the random walk starting from $(k+2,k)$ reaches $(0,0)$.

To prove \eqref{eqn:prop1}, we consider two random walks in $\Z^2$ defined above.  Let $A$ be the random walk starting from $(k+2,k)$, and let $B$ be the random walk starting from $(k+1,k+1)$.  Let $H_A$ and $H_B$ be the event that $A$ and $B$ reaches $(0,0)$ respectively.  To prove (\ref{eqn:prop1}), it is sufficient to show:
$$\Pr(H_A)<\Pr(H_B).$$
To formalize this idea, we define a coupling between $A$ and $B$ such that:  1) whenever $A$ reaches $(0,0)$, $B$ also reaches $(0,0)$, and 2) with a positive probability, $B$ reaches $(0,0)$ but $A$ never does.

In defining the coupling, we use the properties of splitting and merging of Poisson processes~\cite{bertsekas2002introduction}. We reinterpret the random walk by breaking down each \emph{iteration} $i$ into $J(i)$ \emph{steps} such that it is symmetric in the $x$- and $y$-directions (with respect to the line $y = x$) and the movement in each step is ``small''.

\vspace{4pt}

\noindent If at the beginning of iteration $i$ the process is at $(x,y)$ with $x>0$ and $y>0$:
\begin{itemize}
    \item At step~$0$ of iteration $i$, we sample $J(i)\sim \po(2\lambda(i))$, set $(\alpha(i,0), \beta(i,0)) = (-1,-1)$, and update $(x,y)\mapsto (x+\alpha(i,0),y+\beta(i,0))$;
    \item At each step~$j$ for $j=1,\ldots,J(i)$, $(\alpha(i,j), \beta(i,j)) = (1,0)$ with probability $0.5$, and $(\alpha(i,j), \beta(i,j)) = (0,1)$ otherwise.  Update $(x,y)\mapsto (x+\alpha(i,j),y+\beta(i,j))$.\footnote{Standard results from Poisson process indicate that, $\sum_{j =1}^{J(i)} \alpha(i,j)\sim\po(\lambda(i))$, and $\sum_{j =1}^{J(i)} \beta(i,j)\sim\po(\lambda(i))$ which are two independent Poisson random variables.}
\end{itemize}
On the other hand, if $x = 0$ (or $y = 0$) at the beginning of iteration:
\begin{itemize}
    \item At step~$0$ of iteration $i$, we sample $J(i)\sim \po(2\lambda(i))$, set $\big(\alpha(i,0), \beta(i,0)\big) = (0,-1)$ (or $(-1,0)$ if $y=0$), and update $(x,y)\mapsto \big(x+\alpha(i,0),y+\beta(i,0)\big)$;
    \item At each step~$j$ for $j=1,\ldots,J(i)$, with probability $0.5$ $\big(\alpha(i,j), \beta(i,j)\big) = (1,0)$, (or $\big(\alpha(i,j), \beta(i,j)\big) = (0,1)$) and $(\alpha(i,j), \beta(i,j)) = (0,0)$, otherwise.  Update $(x,y)\mapsto \big(x+\alpha(i,j),y+\beta(i,j)\big)$.
\end{itemize}
If at the end of iteration $i$, $(x,y) = (0,0)$ we stop the process.

Notice that we only switch from one type of iteration to the other  if $x=0$ (or $y=0$) at the \emph{end} of an iteration $i$.  Here way say the random walk is stuck to the axis $x = 0$ (or the axis $y= 0$).  If this happens, it will be stuck to this axis forever. Also, notice that in each step we have at most $1$ unit movement.  Also, in steps $j=1,\ldots,J(i)$ the walk can only move further away from both axes $y= 0$ and $x = 0$.

Let $\big(x(i,j),y(i,j)\big)$ be the position of the random walk after iteration $i$ step $j$, and $\big(x(i), y(i)\big)$ be its position at the end of iteration $i$.
Moreover, let $\alpha(i) = \sum_{j = 1}^{J(i)}\alpha(i,j)$ be the net movement in $x$~direction during iteration~$i$ excluding the movement in Step~$0$.
Let $\bar{\alpha}(i) = \alpha(i)+\alpha(i,0)$ be the net movement including movement in iteration $i$. Similarly define $y$-directional movements $\beta(i)=\sum_{j = 1}^{J(i)}\beta(i,j)$ and $\bar{\beta}(i)$.

\subsection{The coupling}\label{sect:prop12}
We want to show that the probability of $A$ reaching the origin is less that of $B$.  To this end, we create a coupling between the two walks, which we outline here.  Fig.~\ref{fig:coupling_a} and Fig.~\ref{fig:coupling_b}  illustrate most aspects of this coupling.
In the description of the coupling, we will let $B$ move ``freely'', and define how $A$ is ``coupled with'' $B$.

Recall that $A$ starts at $(k+2,k)$ and $B$ starts at $(k+1,k+1)$.  At the beginning, we set $A$'s movement to be identical to $B$'s.  Before one of them hits the origin, either of the following two events must happen: $A$ and $B$ become symmetric to the line $x=y$ at some step, $\mathcal{E}_{\textsf{symm}}$,  or $A$ reaches the axis $y = 0$ at the end of some iteration, $\mathcal{E}_{\textsf{skew}}$.  This is called \textsf{Phase I} and is further discussed in Sect.~\ref{sect:phase1}.

In the first case $\mathcal{E}_{\textsf{symm}}$, the positions of $A$ and $B$ are symmetric.  We set $A$'s movement to mirror $B$'s movement.  Therefore, in this case, $A$ and $B$ will both hit the origin, or neither of them will.  This is called \textsf{Phase II Symm} and is further discussed in Sect.~\ref{sect:symm}.

For the latter case $\mathcal{E}_{\textsf{skew}}$, $A$ reaches the axis $y=0$ at iteration $\mathcal{T}_{\textsf{skew}}$.  We call the process is in \textsf{Phase II Skew} and further discussed in Sect.~\ref{sect:skew}. Because $B$ starts one unit above $A$ and one unit to the left of $A$, at iteration $\mathcal{T}_{\textsf{skew}}$, $B$ is at the axis $y = 1$ and one unit to the left of $A$.
Next we couple $A$'s movement in the $x$-direction to be identical to $B$'s, so that $B$ is always one unit to the left of $A$.   This coupling continues unless $B$ hits the axis $x = 0$.  Denote this iteration $\mathcal{T}^*$.  At time $\mathcal{T}^*$, $A$ is one unit to the right of the axis $x = 0$.  Recall that at iteration $\mathcal{T}_{\textsf{skew}}$ when $\mathcal{E}_{\textsf{skew}}$ happens, $B$ is one unit above the axis so that $y = 1$.  Therefore, we can couple the movement of $A$ in the $x$-direction after iteration $\mathcal{T}^*$ with $B$'s movement in the $y$-direction after iteration $\mathcal{T}_{\textsf{skew}}$.  Because $\lambda(i)$ increases with $i$, we can couple the walks in such a way as to ensure that $A$ moves toward the origin at a strictly slower rate than $B$ does.  Therefore, $A$ only reaches the y-axis $x = 0$ if $B$ reaches the x-axis $y = 0$, and we have shown that $A$ is less likely to reach the origin than $B$ does.

Let $\big(x^A(i,j),y^A(i,j)\big)$, and $\big(x^B(i,j),y^B(i,j)\big)$ be the coordinates for $A$ and $B$ respectively after iteration $i$ step $j$.
Similarly, let $J^A(i)$ and $J^B(i)$ be the number of steps for $A$ and $B$ in iteration $i$.
Let $\alpha^A(i,j)$ and $\alpha^B(i,j)$ be the $x$-direction movements of both walks in iteration $i$ step $j$, and $\beta^A(i,j)$ and $\beta^B(i,j)$ be the corresponding $y$-direction movements.

\subsubsection{Phase I}\label{sect:phase1}
Starting with $\big(x^A(0), y^A(0)\big) = (k+2, k)$ and $\big(x^B(0), y^B(0)\big) = (k+1, k+1)$, $A$ moves in exactly the same way as $B$, i.e., $J^A(i) = J^B(i)$, $\alpha^A(i,j)=\alpha^B(i,j)$ and $\beta^A(i,j)=\beta^B(i,j)$, until one of the following two events happens.
    \begin{description}
        \item[Event $\mathcal{E}_{\textsf{symm}}$] The current positions of $A$ and $B$ are symmetric with respect to the line $y=x$, i.e., $x^A(i,j)-x^B(i,j)=y^B(i,j)-y^A(i,j)$ and $x^A(i,j)+x^B(i,j)=y^A(i,j)+y^B(i,j)$. Notice that $\mathcal{E}_{\textsf{symm}}$ may happen in some middle step $j$ of an iteration $i$. When $\mathcal{E}_{\textsf{symm}}$ happens, we move on to \textsf{Phase II Symm}.
        \item[Event $\mathcal{E}_{\textsf{skew}}$] $A$ hits the axis $y= 0$ \emph{at the end of an iteration}. Notice that this means $A$ is then stuck to the axis $y= 0$ forever. When $\mathcal{E}_{\textsf{skew}}$ happens, we move on to \textsf{Phase II Skew}.  Note that $B$ is one unit away from the axis $y= 0$, $y^B = 1$.
        We remark that the in the third part we show, if event $\mathcal{E}_{\textsf{skew}}$ happens, $B$ has a higher chance to reach $(0,0)$ than $A$.
    \end{description}

The following three claims will be useful.
\begin{claim}
$A$ is always below the line $y=x$ before $\mathcal{E}_{\textsf{symm}}$ happens, so $A$ will never hit the axis $x = 0$ in \textsf{Phase I}.
\end{claim}

\begin{proof}
To see this, $A$ can only have four types of movements in each step: lower-left $(x,y)\mapsto(x-1,y-1)$, up $(x,y)\mapsto(x,y+1)$, and right $(x,y)\mapsto(x+1,y)$. It is easy to see that, 1) $A$ will never step across the line $y=x$ in one step, and 2) if $A$ ever reaches the line $y=x$ at $(w,w)$ for some $w$, then $A$ must be at $(w,w-1)$ in the previous step. However, when $A$ is at $(w,w-1)$, $B$ should be at $(w-1,w)$ according to the relative position of $A,B$. In this case event $\mathcal{E}_{\textsf{symm}}$ already happens.
\end{proof}

\begin{claim}
$\mathcal{E}_{\textsf{symm}}$ and $\mathcal{E}_{\textsf{skew}}$ cannot happen simultaneously.
\end{claim}
\begin{proof}
Suppose $\mathcal{E}_{\textsf{symm}}$ and $\mathcal{E}_{\textsf{skew}}$ happen at the same time, then it must be that $A$ is at $(1,0)$ and $B$ is at $(0,1)$, as the relative position of $A$ and $B$ is unchanged in \textsf{Phase I}, and this must be at the end of a certain \emph{iteration}.
In the previous iteration, $A$ must be at $(2,1)$, since $\mathcal{E}_{\textsf{skew}}$ did not happen yet and $A$ is below the line $y=x$.
However, $B$ is at $(1,2)$ when $A$ is at $(2,1)$, implying that case $\mathcal{E}_{\textsf{symm}}$ has already happened in the previous iteration, which is a contradiction.
\end{proof}

\begin{claim}
$B$ cannot reach the axis $x = 0$ before either $\mathcal{E}_{\textsf{symm}}$ or $\mathcal{E}_{\textsf{skew}}$ happen.
\end{claim}

\begin{proof}
If $\mathcal{E}_{\textsf{symm}}$ happens before $\mathcal{E}_{\textsf{skew}}$, $B$ cannot reach the axis $x = 0$ before $\mathcal{E}_{\textsf{symm}}$ as $A$ is always below the line $y=x$ and $B$ is always on the upper-left diagonal of $A$.
If $\mathcal{E}_{\textsf{skew}}$ happens before $\mathcal{E}_{\textsf{symm}}$, $B$ cannot reach the axis $x = 0$ before $\mathcal{E}_{\textsf{skew}}$, or even by the time $\mathcal{E}_{\textsf{skew}}$ happens: by the time $\mathcal{E}_{\textsf{skew}}$ happens, $A$ can only at one of $(2,0),(3,0),(4,0),\ldots$ ($A$ cannot be at $(1,0)$, for otherwise $\mathcal{E}_{\textsf{symm}}$ and $\mathcal{E}_{\textsf{skew}}$ happen simultaneously, which is impossible as shown just now), in which case $B$ will not be at the axis $x = 0$.
\end{proof}

\begin{figure}
\centerline{\includegraphics[height=4.5cm]{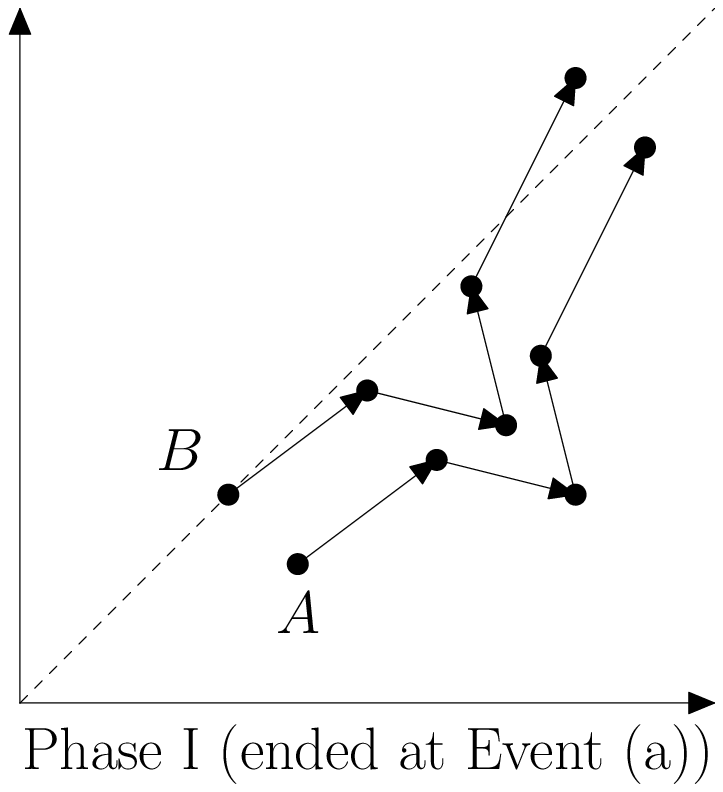}\hspace{1cm}\includegraphics[height=4.5cm]{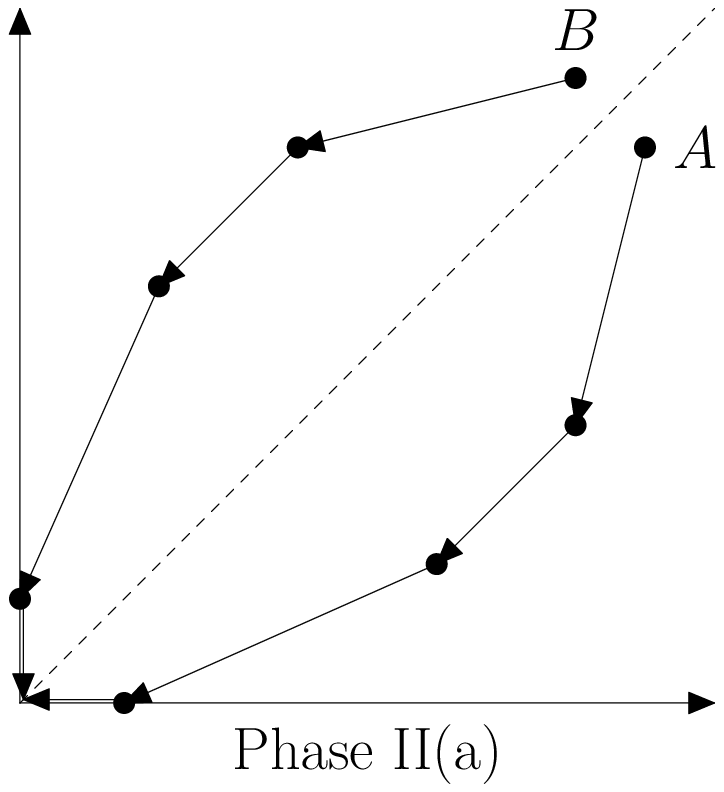}}
\caption{The coupling with \textsf{Phase I} ended at Event $\mathcal{E}_{\textsf{symm}}$}
\label{fig:coupling_a}
\end{figure}

\subsubsection{Phase II Symm}\label{sect:symm}
Let $A$ move in a way that is symmetric to $B$ with respect to the line $y=x$: $J^A(i) = J^B(j)$, $\alpha^A(i,j)=\beta^B(i,j)$ and $\beta^A(i,j)=\alpha^B(i,j)$.
Notice that, in \textsf{Phase II Symm}, $A$ may cross the line $y=x$, after which $A$ is above the line $y=x$ while $B$ is below.
\subsubsection{Phase II Skew}\label{sect:skew}
If event $\mathcal{E}_{\textsf{skew}}$ happens, we need a more complicated coupling.
Suppose \textsf{Phase II Skew} starts after iteration $\mathcal{T}_{\textsf{skew}}$.
Here we use $\mathcal{T}_{S}^A$ ( and $\mathcal{T}_{S}^B$) to denote the hitting time of $A$ (and $B$) to a set of states $S$ which is the first iteration of the process into the set $S\subseteq \Z^2$.
For example $i = \mathcal{T}^B_{y = 1}$ is the hitting time of $B$ such that $y^B(i) = 1$.  Here we list six relevant hitting times and their relationship.  $$\mathcal{T}_{\textsf{skew}} =  \mathcal{T}^B_{y = 1} = \mathcal{T}^A_{y = 0} <\mathcal{T}^B_{y = 0}\text{, and }\mathcal{T}_{\textsf{skew}} <\mathcal{T}^B_{x = 0} = \mathcal{T}^A_{x = 1}<\mathcal{T}^A_{x = 0}.$$

Recall that we have defined the coupling such that $B$ moves freely.
For $A$, we first let the $x$-direction movement of $A$ be the same with that of $B$. To be specific,  in each iteration $\mathcal{T}_{\textsf{skew}}<i\le \mathcal{T}^B_{x = 0}$, set $J^A(i) = J^B(i)$.  At step $j$, we set  $\alpha^A(i,j)=\alpha^B(i,j)$ and $\beta^A(i,j) = 0$ ($\beta^A(i,j)$ is always $0$ now, as $A$ is stuck to the axis $y= 0$).
Till now, the relative position of $A$ and $B$ in $x$-coordinate is preserved $x^A(i,j)=x^B(i,j)+1$.
Let $\mathcal{E}^*$ be the event that $B$ reaches the axis $x = 0$, and let $\mathcal{E}^*$ happens at the end of iteration $\mathcal{T}^* = \mathcal{T}_{x = 0}^B$.  We further define $\Delta = \mathcal{T}^* - \mathcal{T}_{\textsf{skew}}$ to be the additional time before $x^B = 0$  (if both stopping times exist), and $L = \mathcal{T}^B_{y = 0} - \mathcal{T}_{\textsf{skew}}$  to be the additional time before $y^B = 0$ (if both stopping times exist).

At the end of iteration $\mathcal{T}^*$, the positions for $A$ is one unit to the right of the origin.  That is $x^A(\mathcal{T}^*) = 1$ while $y^A(\mathcal{T}^*)\big)= 0$.
Informally, we want to couple the movement of $A$ from $(1,0)$ at $\mathcal{T}^*$ to the movement of $B$ in the $y$-direction at $\mathcal{T}_{\textsf{skew}}$ which is one unit above the axis at $y = 1$.
Formally, starting at $(1,0)$, $A$ is a $1$-dimensional random walk on the axis $y= 0$, and we couple it to $B$ in the following way.
\begin{itemize}
    \item For each $t=1,\ldots,L$, we couple $A$'s movement in the $x$~direction at iteration $\mathcal{T}^*+t$ with $B$'s movement $\Delta$ steps earlier in the $y$~direction at iteration $\mathcal{T}^*+t - \Delta = \mathcal{T}_{\textsf{skew}}+t$ such that $\alpha^A(\mathcal{T}^*+t)\sim\po(\lambda(\mathcal{T}^*+t))$ and  $\alpha^A(\mathcal{T}^*+t)\ge \beta^B(\mathcal{T}_{\textsf{skew}}+t)$.
    \footnote{Here is an example of such a coupling. Consider iteration $i = \mathcal{T}^*+t$ for $A$, and we want to couple it with $B$'s movement at iteration $\iota = \mathcal{T}_{\textsf{skew}}+t$.  Let $J^B(\iota)$ be the number of steps of $B$ in the iteration $\iota$ which is not necessary equal to the number of steps of $A$ after iteration $\mathcal{T}^*$.  At step $0$, we sample a non-negative integer $d(i)\sim \po(2(\lambda(i)-\lambda_{\iota}))$ independent to $J^B(\iota)$, and set the number of steps of $A$ to be $J^A(i) = J^B(\iota)+d(i)$.  Then set $\alpha^A(i,0) = -1$ and $\beta(i,0)^A = 0$.  At each step $j = 1,\ldots, J^B(\iota)$, we set $(\alpha^A(i,j), \beta^A(i,j)) = (\beta_{\iota j}^B, 0)$. At the later steps $j = J^B(\iota)+1, \ldots, J^A(i)$, we set $(\alpha^A(i,j), \beta^A(i,j)) = (1,0)$ with probability $0.5$, or $(0,0)$ otherwise.}
    \item We do not couple $A$ to $B$ for future iterations after $\mathcal{T}^*+L$.
\end{itemize}
A key property of this coupling is that the $x$-coordinate of $A$ at $\mathcal{T}^*+t$ is always greater or equal to the $y$-coordinate of $B$ at iteration $\mathcal{T}_{\textsf{skew}}+t$.
\begin{claim}
For all $t = 1, \ldots, L$,
$x^A(\mathcal{T}^*+t)\ge y^B(\mathcal{T}_{\textsf{skew}}+t)$.
\end{claim}
\begin{proof}
We use induction.  For the base case, we have $1 = x^A(\mathcal{T}^*) = y^B(\mathcal{T}_{\textsf{skew}})$ from the definitions of $\mathcal{T}_{\textsf{skew}}$ and $\mathcal{T}^*$.  For the inductive case, $\alpha^A(\mathcal{T}^*+t)\ge \beta^B(\mathcal{T}_{\textsf{skew}}+t)$ due to our coupling.
\end{proof}

\begin{figure}
\centerline{\includegraphics[height=4cm]{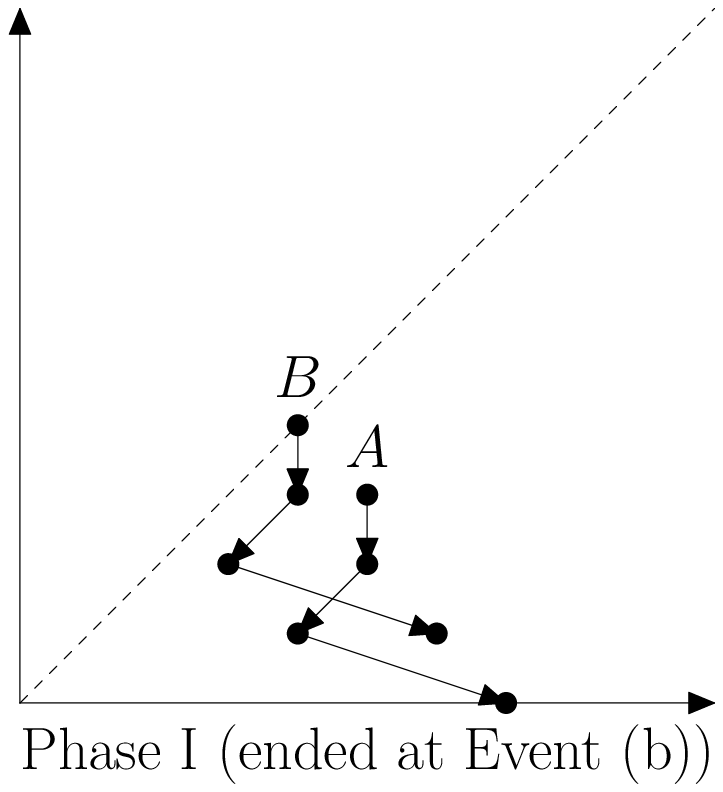}\hspace{1cm}\includegraphics[height=4cm]{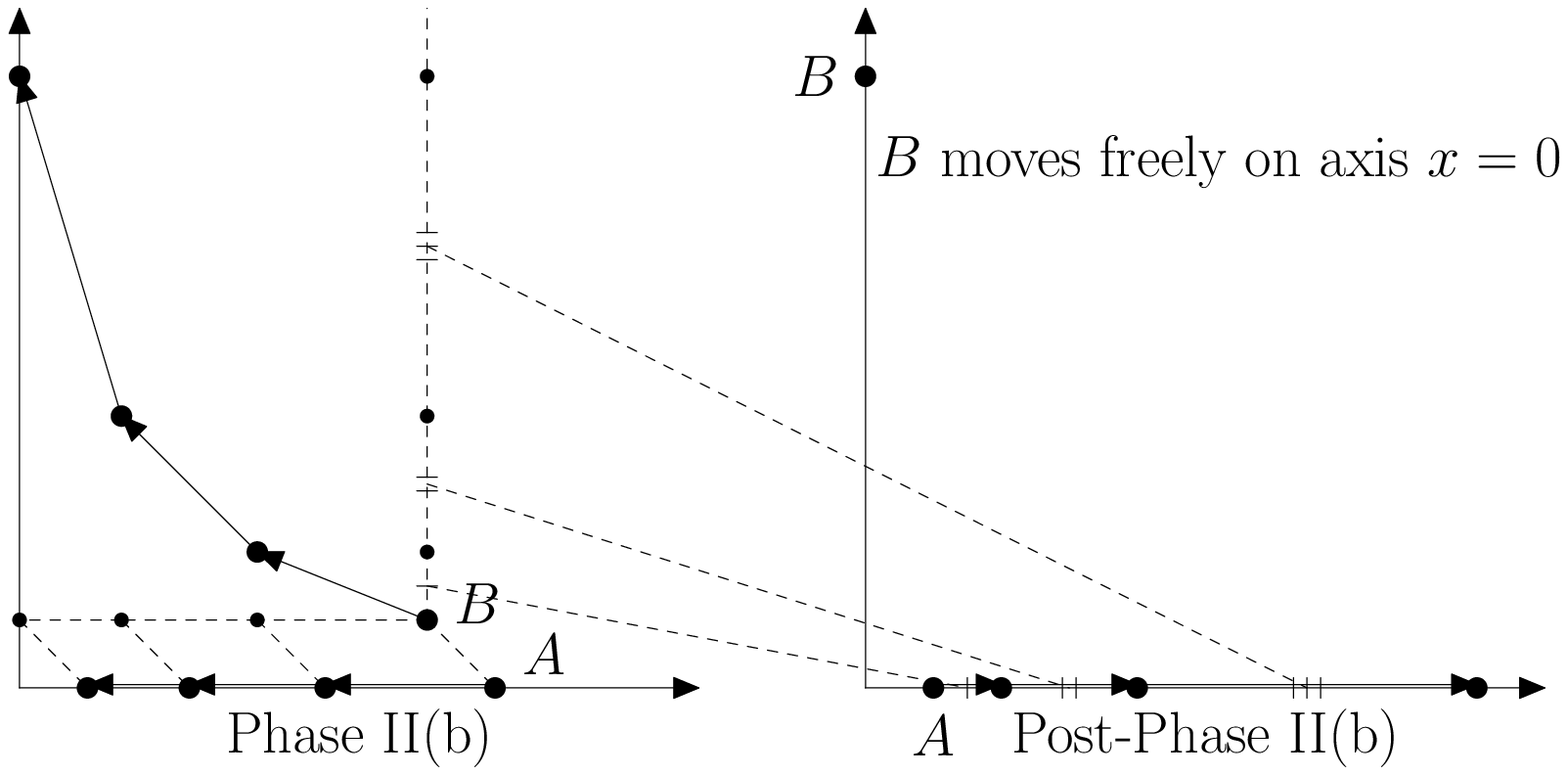}}
\caption{The coupling with \textsf{Phase I} ended at Event $\mathcal{E}_{\textsf{skew}}$, if $\mathcal{E}^*$ happens}
\label{fig:coupling_b}
\end{figure}

\subsection{Validity of the coupling}\label{sect:prop13}
The coupling induces the correct marginal random walk process for $B$, as we have defined the coupling in a way that $B$ is moving ``freely'' and $A$ is being ``coupled'' with $B$.
The only non-trivial part is to show that the coupling induces the correct marginal random walk process for $A$.  It is straightforward to check that the marginal probabilities are correct, before the event $\mathcal{E^*}$ occurs, or if the event $\mathcal{E^*}$ does not occur.
If $\mathcal{E}^*$ happens (which implies that the process enters \textsf{Phase II Skew} and $B$ reaches the axis $x = 0$), the movement of $A$ in the $x$~direction is coupled with $B$'s movement in $y$~direction $\Delta = \mathcal{T}^*-\mathcal{T}_{\textsf{skew}}$ iterations ago.
We note that $B$'s movements in the $x$~direction and the $y$~direction are independent and $A$ does not contain two iterations that are coupled to a same iteration of $B$.
Therefore, the movements of $A$ in $x$~direction after $\mathcal{T}^*$ are independent to its previous movement, so the marginal distribution is correct.
Fig.~\ref{fig:coupling_validity} illustrates the coupling time line.

\begin{figure}[ht]
\centerline{\includegraphics{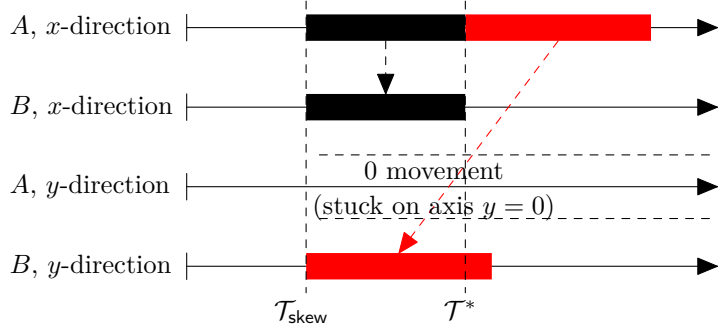}}
\caption{The time line for the coupling after event $\mathcal{E}_{\textsf{skew}}$ happens.}
\label{fig:coupling_validity}
\end{figure}

\begin{remark}
The coupling of the two random walks $A$ and $B$ in $\mathbb{Z}^2$ in the proof above can be alternatively viewed as a coupling of four independent random walks in $\mathbb{Z}$ (this is why we have said that ``we simultaneously couple four cascade processes'' in the introduction), as the $x$-directional and $y$-directional movements for both $A$ and $B$ correspond to the four terms in inequality~\eqref{eqn:prop1}, which are intrinsically independent.
\end{remark}

\subsection{Proof of Inequality~(\ref{eqn:prop1})}\label{sect:prop14}
It suffices to show that in our coupling $H_A\subseteq H_B$ and $H_B\setminus H_A$ has a positive probability, because this implies inequality~\eqref{eqn:prop1}:
$\Pr(H_A) = \Pr(H_B\cap H_A) < \Pr(H_B\cap H_A) +\Pr(H_B\setminus H_A) = \Pr(H_B).$
We aim to show the following:
\begin{enumerate}
    \item \label{item:0} if the coupling never moves to \textsf{Phase II}, neither $A$ nor $B$  reaches $(0,0)$;
    \item \label{item:1} if the coupling moves to \textsf{Phase II Symm}, $A$ reaches $(0,0)$ if and only if $B$ reaches $(0,0)$;
    \item \label{item:2} if the coupling moves to \textsf{Phase II Skew}, $A$ reaches $(0,0)$ implies that $B$ also reaches $(0,0)$;
    \item \label{item:3} there is an event with a positive probability such that $B$ reaches $(0,0)$ but $A$ does not.
\end{enumerate}
The first, second, and third show $H_A\subseteq H_B$.  The last one shows $H_B\setminus H_A$ has a positive probability.

\ref{item:0} is trivial. \ref{item:1} follows from symmetry.

To see \ref{item:2}, first notice that in  \textsf{Phase II Skew}, $\mathcal{E}^*$ must happens if $A$ ever reaches $(0,0)$: because $A$ can move to the left by at most $1$ unit in each iteration, $A$ must first reach $(1,0)$, but at this point $x^B = 0$ and event $\mathcal{E}^*$ happens.  Now consider the case that $B$ never reaches the origin after event $\mathcal{E}^*$.  Then the $x$ movement of $A$ remains coupled to the $y$-movement of $B$ in such a way that $\bar{\alpha}^A(\mathcal{T}^* + t) \geq \bar{\beta}^B(\mathcal{T}_{\textsf{skew}} + t)$.  Walk $A$ starts at $x^A = 1$, and walk $B$ starts at $y^B = 1$.  Therefore,  $A$ cannot reach the origin if $B$ does not.   In the case walk $B$ meets the origin, the statement is vacuously true.

For \ref{item:3}, to show $\Pr(H_B\setminus H_A)>0$, we define the following event which consists of four parts. i) For all $i = 1, \ldots, k$, it happens that $\alpha^A(i)=\beta^A(i)=0$, in which case the event $\mathcal{E}_{\textsf{skew}}$ happens at $\mathcal{T}_{\textsf{skew}}=k$ and $A$ reaches $(2,0)$.
ii) For $i=k+1$, it happens that $\alpha^A(i) = 0$ and $\beta^B(i) = 1$, in which case $A$ reaches $(1,0)$ and $B$ reaches $(0,1)$, and the process $B$ reaches the axis $x = 0$  at iteration $\mathcal{T}^* = k+1$. iii) In iteration $i = \mathcal{T}^*+1$, it happens that $\beta^B(i) = 0$, so $B$ reaches $(0,0)$.  On the other hand, by the coupling $\alpha^A(\mathcal{T}^*+1) \ge \beta^B(\mathcal{T}_{\textsf{skew}}+1) = 1$, so $A$ does not reach $(0,0)$ at iteration $\mathcal{T}^*+1 = k+2$.  iv) Finally, it happens that $\alpha^A(i) \ge 1$ for all $i>k+2$.
It is straightforward the i), ii), and iii) happen with positive probabilities.  By direct computations, iv) happens with a positive probability as well.\footnote{The event that  $\alpha^A(i) \ge 1$ for all $i>k+2$ happens with probability $\prod_{i> k+2} \Pr(\po(\lambda(i))\ge 1) = \prod_{i> k+2} (1-\exp(-\lambda(i)))\ge \prod_{i\ge r+1} (1-\exp(-\binom{i-1}{r-1}c^r))$ which is a positive constant depending on $r$ and $c$.}
Since the above event consisting of i), ii), iii) and iv) belongs to $H_B\setminus H_A$ and each of the four sub-events happens with a positive probability, \ref{item:3} is implied.

From \ref{item:1}, \ref{item:2}, and \ref{item:3}, we learn that the probability that $B$ reaches $(0,0)$ is strictly larger than that of $A$, which implies inequality (\ref{eqn:prop1}) and concludes the proof.

\section{Optimal Seeds in Submodular \infmax}
\label{sect:submodular}
We have seen that putting all the $K$ seeds in a single leaf is optimal for $r$-complex contagion, when the root node has weight $\omega(1/n^{1+1/r})$.
To demonstrate the sharp difference between $r$-complex contagion and a submodular cascade model, we present a submodular \infmax example where the optimal seeding strategy is to put no more than one seed in each leaf.
The hierarchy tree $T$ in our example meets all the assumptions we have made in the previous sections, including large communities, proper separation, and $w(\mbox{root})=\omega(1/n^{1+1/r})$, where $r$ is now an arbitrarily fixed integer with $r\geq2$.

We consider a well-known submodular cascade model, \emph{the independent cascade model}~\cite{kempe2003maximizing}, where, after seeds are placed, each edge $(u,v)$ in the graph appears with probability $p_{uv}$ and vertices in all the connected components of the resultant graph that contain seeds are infected.
In our example, the probability $p_{uv}$ is the same for all edges, and it is $p=1/n^{1-\frac1{4r}}$.
The hierarchy tree $T$ contains only two levels: a root and $K$ leaves.
The root has weight $1/n^{1+\frac1{2r}}$, and each leaf has weight $1$.
After $G\sim\calG(n,T)$ is sampled and each edge in $G$ is sampled with probability $p$, the probability that an edge appears between two vertices from different leaves is $(1/n^{1-\frac1{4r}})\cdot (1/n^{1+\frac1{2r}})=o(1/n^2)$, and the probability that an edge appears between two vertices from a same leaf is $1\cdot(1/n^{1-\frac1{4r}})=\omega(\log n/n)$.
Therefore, with probability $1-o(1)$, the resultant graph is a union of $K$ connected components, each of which corresponds to a leaf of $T$.
It is then straightforward to see that the optimal seeding strategy is to put a single seed in each leaf.

\section{A Dynamic Programming Algorithm}
In this section, we present an algorithm which finds an optimal seeding strategy when all $w(t)$'s fall into two regimes: $w(t)=\omega(1/n^{1+1/r})$ and $w(t)=o(1/n^2)$.
We will assume this for $w(t)$'s throughout this section.
Since a parent tree node always has less weight than its children (see Definition~\ref{defi:SHB}), we can decompose $T$ into \emph{the upper part} and \emph{the lower part}, where the lower part consists of many subtrees whose roots have weights in $\omega(1/n^{1+1/r})$, and the upper part is a single tree containing only tree nodes with weights in $o(1/n^2)$ and whose leaves are the parents of those roots of the subtrees in the lower part.
We call each subtree in the lower part a \emph{maximal dense subtree} defined formally below.

\begin{definition}\label{defi:densesubtree}
Given a hierarchy tree $T=(V_T,E_T,w,v)$, a subtree rooted at $t\in V_T$ is a \emph{maximal dense subtree} if  $w(t)=\omega(1/n^{1+1/r})$, and either $t$ is the root, or $w(t')=O(1/n^{1+1/r})$ where  $t'$ is the parent of $t$.
\end{definition}

Since we have assumed either $w(t)=\omega(1/n^{1+1/r})$ or $w(t)=o(1/n^2)$, $w(t')=O(1/n^{1+1/r})$ in the definition above implies $w(t')=o(1/n^2)$.

The idea of our algorithm is the following: firstly, after the decomposition of $T$ into the upper and lower parts, we will show that the weights of the tree nodes in the upper part, falling into $w(t)=o(1/n^2)$, are negligible so that we can treat the whole tree $T$ as a forest with only those maximal dense subtrees in the lower part (that is, we can remove the entire upper part from $T$); secondly, Theorem~\ref{THM:MAIN} shows that after we have decide the number of seeds to be allocated to each maximal dense subtree, the optimal seeding strategy is to put all the seeds together in a single leaf that has the highest density  defined in Definition~\ref{defi:density}; finally, we use a dynamic programming approach to allocate the $K$ seeds among those maximal dense subtrees.

Now, we are ready to describe our algorithm, presented in Algorithm~\ref{alg}.

\begin{algorithm}
  \caption{The \infmax algorithm}\label{alg}
  \begin{algorithmic}[1]
    \STATE \textbf{Input: } $r\in\mathbb{Z}$ with $r\geq2$, $T=(V_T,E_T,w,v)$, and $K\in\mathbb{Z}^+$
    \STATE Find all maximal dense subtrees $T_1,\ldots,T_m$, and let $r_1,\ldots,r_m$ be their roots (Definition~\ref{defi:densesubtree}).
    \STATE For each $T_i$ and each $k=0,1,\ldots,K$, let $\bm{s}_i^\ast(k)$ be the seeding strategy that puts $k$ seeds in the leaf $t\in L_{T_i}$ with the highest density, and let
    $$h(T_i,k)=\lim_{n\rightarrow\infty}\frac{\E_{G\sim\calG(v(r_i)\cdot n, T_i)}[\sigma_{r,G}(\bm{s}_i^\ast(k))]}{n}$$
    be the expected number of infected vertices in the subgraph defined by $T_i$, divided by the total number of vertices in the whole graph.
    \STATE Let $S[i,k]$ store a seeding strategy that allocates $k$ seeds in the first $i$ subtrees $T_1,\ldots,T_i$, and let $H[i,k]$ be the expected total number of infected vertices corresponding to $S[i,k]$, divided by $n$.
    \FOR{$k=0,1,\ldots,K$}
    \STATE set $S[1,k]=\bm{s}_1^\ast(k)$ and $H[1,k]=h(T_1,k)$.
    \ENDFOR
    \FOR{each $i=2,\ldots,m$}
    \FOR{$k=0,1,\ldots,K$}
        \STATE $\displaystyle k_i = \argmax_{k_i \in \{0, 1, \ldots, k\}} H[i-1,k-k_i]+h(T_i,k_i)$;
        \STATE set $S[i,k]$ be the strategy that allocates $k-k_i$ seeds among $T_1,\ldots,T_{i-1}$ according to $S[i-1,k-k_i]$ and puts the remaining $k_i$ seeds in the leaf of $T_i$ with the highest density;
        \STATE set $H[i,k]=H[i-1,k-k_i]+h(T_i,k_i)$;
    \ENDFOR
    \ENDFOR
    \STATE\textbf{Output:} the seeding strategy $S[m,K]$.
  \end{algorithmic}
\end{algorithm}

The correctness of Algorithm~\ref{alg} follows immediately from Theorem~\ref{thm:isolated} (below) and Theorem~\ref{THM:MAIN}.  Theorem~\ref{thm:isolated} shows that we can ignore the upper part of $T$ and treat $T$ as the forest consisting of all the maximal dense subtrees of $T$ when considering the \infmax problem.
Recall Theorem~\ref{THM:MAIN} shows that for each subtree $T_i$ and given the number of seeds, the optimal seeding strategy is to put all the seeds on the leaf with the highest density.

\begin{theorem}\label{thm:isolated}
Given $T=(V_T,E_T,w,v)$, let $\{T_1,\ldots,T_m\}$ be the set of all $T$'s maximal dense subtrees and let $T^-$ be the forest consisting of $T_1,\ldots,T_m$. For any seeding strategy $\bm{k}$ and any $r\geq2$, we have $\Sigma_{r,T}(\bm{k})=\Sigma_{r,T^-}(\bm{k})$.
\end{theorem}
\begin{proof}
Since the total number of possible edges between $T^-$ and the rest of the tree is upper bounded by $n^2$ and each such edge appears with probability $o(1/n^2)$, the expected number of edges is $o(1)$.
By Markov's inequality the probability there exists edges between $T^-$ and the rest of the tree $o(1)$.
Therefore, we have
$$ \frac{\E\limits_{G\sim\calG(n,T)}\left[\sigma_{r,G}(\bm{k})\right]}n=\frac{o(1) O(n)+(1-o(1)) \E\limits_{G\sim\calG(n,T^-)}\left[\sigma_{r,G}(\bm{k})\right]}n.$$
Taking $n\rightarrow\infty$ we have concludes the proof.
\end{proof}

Finally, it is straightforward to see the time complexity of Algorithm~\ref{alg}, in terms of the number of evaluations of $\Sigma_{r,\calG(n,T)}(\cdot)$.

\begin{theorem}
Algorithm~\ref{alg} requires $O_I(|V_T|K^2)$ computations of $\Sigma_{r,\calG(n,T)}(\cdot)$.
\end{theorem}

\section{Conclusion and Future Work}

In this paper, we presented an influence maximization algorithm which finds optimal seeds for the stochastic hierarchical blockmodel, assuming the weights of tree nodes do not fall into a narrow regime between $\Omega(1/n^2)$ and $O(1/n^{1+1/r})$.
As a crucial observation behind the algorithm, when the root of the tree has weight $\omega(1/n^{1+1/r})$, our results show that the optimal seeding strategy is to put all the seeds together.
Our results provide a formal verification for the intuition that one should put the seeds close to each other to maximize the synergy effect in a nonsubmodular cascade model.

\paragraph{Removing Limitations}
One obvious future direction is to extend our algorithm such that it works for weights of tree nodes between $\Omega(1/n^2)$ and $O(1/n^{1+1/r})$ as well.
Related to this, \citet{schoenebeck2017beyond} shows that \infmax for the complex contagion on the stochastic hierarchical blockmodel is NP-hard to approximate to within factor $n^{1-\varepsilon}$ if vertices have non-homogeneous thresholds, i.e., each vertex $v$ has a individual threshold $r_v\in\mathbb{Z}^+$ such that $v$ is infected when it has at least $r_v$ infected neighbors.
It is unknown whether this inapproximability result carries over to the homogeneous case where all agents have the same threshold.

It is also interesting to see if our main result Theorem~\ref{THM:MAIN} still holds without the proper separation assumption.
We only use this assumption in the proof of Proposition~\ref{prop:isolate}.
To remove the proper separation assumption, more insight is needed on the behavior of the cascade in the critical leaves.
As a next step for this, one might consider the case when leaves $t_1$ and $t_2$ have weights $c_1n^{-1/r}$ and $c_2n^{-1/r}$ respectively, and their parent $t$ has weight $dn^{-1/r}$ with $d<c_1$ and $d<c_2$; it is an interesting open problem to see that if it is still optimal to either put all the seeds in $t_1$ or to put all the seeds in $t_2$. We conjecture this is true.

\paragraph{Extension}
One way to extend our results is to relax the assumption that the network is known.  For example, can the network be learned from observing previous cascades, or by experimenting with them?  Or, can they be elicited from agents with limited, local knowledge?  Another direction would be to leverage these results to create heuristics that work well on real-world networks.  A final direction would be more careful empirical studies (particularly experiments) about the nature of various cascades (e.g. submodular versus nonsubmodular).

\bibliographystyle{plainnat}
\bibliography{reference}

\newpage
\appendix
\section{Proof of Lemma~\ref{lemmaA}}
The proof will follow the structure of the proof sketch in the main body of this paper.

Let $E$ be the event that at least one leaf (or tree node) is activated at the end of the cascade.
By our definition, $P_{\bm{k}}=\lim_{n\rightarrow\infty}\Pr(E)$.
Given a seeding strategy $\bm{k}$, let $\sigma(\bm{k}):=\E_{G\sim\calG(n,T)}[\sigma_{r,G}(\bm{k})]$ be the expected number of infected vertices, $\sigma(\bm{k}\mid E):=\E_{G\sim\calG(n,T)}[\sigma_{r,G}(\bm{k})\mid E]$ be the expected number of infected vertices conditioning on event $E$, and $\sigma(\bm{k}\mid \neg E):=\E_{G\sim\calG(n,T)}[\sigma_{r,G}(\bm{k})\mid \neg E]$ be the expected number of infected vertices conditioning on that $E$ does not happen.
We have
$$\sigma(\bm{k})=\Pr(E)\cdot\sigma(\bm{k}\mid E)+\left(1-\Pr(E)\right)\cdot\sigma(\bm{k}\mid \neg E),$$
and
\begin{equation}\label{eqn:lemmaA1}
\Sigma_{r,T}(\bm{k})=\lim_{n\rightarrow\infty}\frac{\sigma(\bm{k})}{n}=P_{\bm{k}}\cdot\lim_{n\rightarrow\infty}\frac{\sigma(\bm{k}\mid E)}n+\left(1-P_{\bm{k}}\right)\cdot\lim_{n\rightarrow\infty}\frac{\sigma(\bm{k}\mid \neg E)}n.
\end{equation}
To prove Lemma~\ref{lemmaA}, it is sufficent to show the following two claims:
\begin{enumerate}
    \item First, we show that $1-P_{\bm{k}}>0$ implies $\sigma(\bm{k}\mid \neg E)=o(n)$, so the second term in (\ref{eqn:lemmaA1}) is always $0$ (Sect.~\ref{sect:lemmaAc1}).
    \item Second, to conclude the proof, it suffices to show that $\sigma(\bm{k}\mid E)=cn+o(n)$ for some constant $c$ which does not depend on $\bm{k}$, which implies that the first term in (\ref{eqn:lemmaA1}) is monotone in $P_{\bm{k}}$ (Sect.~\ref{sect:lemmaAc2}).
\end{enumerate}
These two claims correspond to the second and the third paragraphs in the sketch of the proof.

The following proposition is useful for proving both claims.
\begin{proposition}\label{prop:constantOutsideInfection}
Suppose the root of $T$ has weight $\omega(1/n^{1+1/r})$ and consider a leaf $t$. If there are $\Theta(n)$ infected vertices in $V\setminus V(t)$, then these infected vertices outside $V(t)$ will infect $\omega(1)$ vertices in $V(t)$ with probability $1-o(1)$.
\end{proposition}
\begin{proof}
Let $X=\Theta(n)$ be the number of infected vertices in $V\setminus V(t)$.
For each $u\in V(t)$ and $v\in V\setminus V(t)$, we assume that the probability $p_{uv}$ that the edge $(u,v)$ appears satisfies $p_{uv}=\omega(1/n^{1+1/r})$ and $p_{uv}=o(1/n)$, where $p_{uv}=\omega(1/n^{1+1/r})$ holds since the root of $T$ has weight $\omega(1/n^{1+1/r})$, and assuming $p_{uv}=o(1/n)$ may only decrease the number of infected vertices in $V(t)$ if the least common ancestor of the two leaves containing $u$ and $v$ has weight $\Omega(1/n)$.
Let $p$ be the minimum probability among those $p_{uv}$'s, and we further assume that each edge $(u,v)$ appears with probability $p$, which again may only reduce the number of infected vertices in $V(t)$.

For each vertex $u\in V(t)$, by only accounting for the probability that it has exactly $r$ neighbors among those $X$ outside infected vertices, the probability that $u$ is infected is at least
$$\rho:=\binom{X}{r}p^r(1-p)^{X-r}=\omega\left( n^r\cdot\left(\frac1{n^{1+1/r}}\right)^r\left(1-\frac1n\right)^n\right)=\omega\left(\frac1n\right),$$
and the expected number of infected vertices in $V(t)$ is at least $v(t)n\cdot\rho=\omega(1)$.

Let $Y$ be the number of vertices in $V(t)$ that are infected due to the influence of $V\setminus V(t)$, so we have $\E[Y]=v(t)n\rho$.
Applying Chebyshev's inequality,
$$\Pr\left(Y\leq\frac12v(t)n\rho\right)\leq\Pr\left(|Y-\E[Y]|\geq\frac12v(t)n\rho\right)$$
$$\qquad\leq\frac{\mbox{Var}(Y)}{(\frac12v(t)n\rho)^2}=\frac{v(t)n\rho(1-\rho)}{\frac14v(t)^2n^2\rho^2}=o(1),$$
where we have used the fact that $n\rho=\omega(1)$ and the variance of the Binomial random variable with parameter $n,p$ is $np(1-p)$.
Therefore, with probability $1-o(1)$, the number of infected vertices in $V(t)$ is at least $\frac12v(t)n\rho=\omega(1)$.
\end{proof}

\subsection{Proof of the First Claim}\label{sect:lemmaAc1}
We consider two cases: 1) $T$ contains no critical or supercritical leaf; 2) $T$ contains at least one critical or supercritical leaf.

If there is no critical or supercritical leaf in $T$, given that the total number of seeds $K=\Theta(1)$ is a constant, Theorem~\ref{thm:gnp_sub} shows that, with high probability, there can be at most $2K=\Theta(1)$ infected vertices even without conditioning on that $E$ has not happened.
To be specific, we can take the maximum weight $w^\ast(t)$ over all the leaves, and assume the entire graph is the Erd\H{o}s-R\'{e}nyi graph $\calG(n,w^\ast(t))$.
This makes the graph denser, so the expected number of infected vertices increases.
We further assume that we have not conditioned on $\neg E$, this further increases the expected number of infected vertices.
However, even under these assumptions, Theorem~\ref{thm:gnp_sub} implies that the total number of infected vertices is less than $2K$ with high probability.
Thus, $\sigma(\bm{k}\mid\neg E)=o(n)$ even without assuming $1-P_{\bm{k}}>0$.

Suppose there is at least one critical or supercritical leaf, and $\Pr(\neg E)=\Theta(1)$ (equivalently, $1-P_{\bm{k}}>0$, as given in the statement of the first claim).
To show that $\sigma(\bm{k}\mid\neg E)=o(n)$, it suffices to show that, conditioning on  there being $\Theta(n)$ infected vertices, $E$ happens with probability $1-o(1)$.
This is because, if $\Pr(\neg E)=\Theta(1)$ and $\Pr(\neg E\mid\sigma(\bm{k})=\Theta(n))=o(1)$, then
$$\Pr\left(\sigma(\bm{k})=\Theta(n)\mid\neg E\right)=\frac{\Pr(\sigma(\bm{k})=\Theta(n))\cdot\Pr(\neg E\mid\sigma(\bm{k})=\Theta(n))}{\Pr(\neg E)}=o(1),$$
which implies $\sigma(\bm{k}\mid\neg E)=o(n)$.

Now, suppose there are $\Theta(n)$ infected vertices;
to conclude the claim, we will show that $E$ happens with probability $1-o(1)$.
Since the number of leaves is a constant, there exists $t'\in L_T$ such that the number of infected vertices in $V(t')$ is $\Theta(n)$.
Let $t$ be a critical or supercritical leaf (we have supposed there is at least one critical or supercritical leaf).
Theorem~\ref{thm:gnp_sup} and Corollary~\ref{cor:gnp_boundary} indicate that, with probability $1-o(1)$, the number of infected vertices in $V(t)$ is either a constant or $v(t)n$.
Therefore, if $t'=t$, with probability $1-o(1)$, those $\Theta(n)$ infected vertices in $V(t)$ will activate $t$, so $E$ happens with probability $1-o(1)$.
If $t'\neq t$, let $X=\Theta(n)$ be such that with probability $1-o(1)$ the number of infected vertices in $V(t')$ is more than $X$, then the total number of vertices in $V(t)$ that are infected by those $X$ vertices in $V(t')$ is $\omega(1)$ (with high probability) according to Proposition~\ref{prop:constantOutsideInfection}.
Theorem~\ref{thm:gnp_sup} and Corollary~\ref{cor:gnp_boundary} show that, with high probability, those $\omega(1)$ infected vertices in $V(t)$ will further spread and activate $t$, which again says that $E$ happens with probability $1-o(1)$.

\subsection{Proof of the Second Claim}\label{sect:lemmaAc2}
As an intuitive argument, Proposition~\ref{prop:constantOutsideInfection}, Theorem~\ref{thm:gnp_sup}, and Corollary~\ref{cor:gnp_boundary} show that, when $E$ happens, with high probability, a single activated leaf will activate all the critical and supercritical leaves, and the number of vertices corresponding to all the critical and supercritical leaves is fixed and independent of $\bm{k}$; based on the tree structure and the number of infected outside vertices, the number of infected vertices in a subcritical leaf may vary; however, we will see that the seeding strategy $\bm{k}$, adding only a constant number of infections, is too weak to significantly affect the number of infected vertices in a subcritical leaf.

To break it down, we first show that all critical and supercritical leaves will be activated with high probability if $E$ happens.
This is straightforward: Proposition~\ref{prop:constantOutsideInfection} shows that an activated leaf can cause $\omega(1)$ infected vertices in every other leaf with high probability, and Theorem~\ref{thm:gnp_sup} and Corollary~\ref{cor:gnp_boundary} indicate that those critical and supercritical leaves will be activated by those $\omega(1)$ infected vertices with high probability.

Lastly, assuming all critical and supercritical leaves are activated, we show that the number of infected vertices in any subcritical leaf does not significantly depend on $\bm{k}$.
We do not need to worry about those seeds that are put in the critical or supercritical leaves, as all vertices in those leaves will be infected later.
As a result, we only need to show that a constant number of seeds in subcritical leaves has negligible effect to the cascade.

We say a subcritical leaf $t$ is \emph{vulnerable} if there exists a criticial or supercritical leaf $t'$ such that the least common ancestor of $t$ and $t'$ has weight $\Omega(1/n)$, and we say $t$ is \emph{not-vulnerable} otherwise.
It is easy to see that a vulnerable leaf $t$ will be activated with high probability conditional on $E$, even if no seed is put into it.
Since each $v\in V(t)$ is connected to one of the $v(t')n$ vertices in $V(t')$ with probability $\Omega(1/n)$, the number of infected neighbors of $v$ follows a Binomial distribution with parameter $(v(t')n,p)$ where $p=\Omega(1/n)$.
We only consider $p=\Theta(1/n)$, as there can only be more infected vertices if $p=\omega(1/n)$.
If $p=\Theta(1/n)$, the Binomial distribution becomes a Poisson distribution with a constant mean $\lambda$ for $n\rightarrow\infty$.
In this case, with constant probability $e^{-\lambda}\frac{\lambda^r}{r!}$, $v$ has $r$ infected neighbors. Therefore, $v$ will be infected with constant probability, and $V(t)$ has $\Theta(n)$ vertices that are infected by $V(t')$ outside.
The second part of Theorem~\ref{thm:gnp_sub} shows that, these $\Theta(n)$ infected vertices will further spread and activate $t$ with high probability.
Therefore, the seeds on those vulnerable subcritical leaves have no effect, since vulnerable subcritical leaves will be activated with high probability regardless the seeding strategy.

Let $t_1,\ldots,t_M$ be all the not-vulnerable subcritical leaves.
Suppose we are at the stage of the cascade process where all those critical, supercritical and vulnerable subcritical leaves have already been activated (as they will with probability $1-o(1)$ since we assumed that $E$ has happened) and we are revealing the edges between $V\setminus\bigcup_{m=1}^MV(t_m)$ and $\bigcup_{m=1}^MV(t_m)$ to consider the cascade process in $\bigcup_{m=1}^MV(t_m)$.
For each $i=0,1,\ldots,r-1$ and each $m=1,\ldots,M$, let $\chi_i^m$ be the number of vertices in $V(t_m)$ that have \emph{exactly} $i$ infected neighbors among $V\setminus\bigcup_{m=1}^MV(t_m)$, which can be viewed as a random variable. For each $m=1,\ldots,M$, let $\chi_r^m$ be the number of vertices in $V(t_m)$ that have \emph{at least} $r$ infected neighbors.
If there are $K_m$ seeds in $V(t_m)$, we increase the value of $\chi_r^m$ by $K_m$.
Let $\bm{\chi}^m=(\chi_0^m,\chi_1^m,\ldots,\chi_r^m)$.
Since $(\bm{\chi}^1,\ldots,\bm{\chi}^M)$ completely characterizes the expected number of infected vertices in the subcritical leaves (the expectation is taken over the sampling of the edges within every $V(t_i)$ and between every pair $V(t_i),V(t_j)$), we let $\sigma(\bm{\chi}^1,\ldots,\bm{\chi}^M)$ be the total number of infected vertices in the subcritical leaves, given $(\bm{\chi}^1,\ldots,\bm{\chi}^M)$.
We aim to show that \emph{adding $K_1,\ldots,K_M$ seeds in $V(t_1),\ldots,V(t_M)$ only changes the expected number of infected vertices by $o(n)$.}

Let $(\bm{\chi}^1,\ldots,\bm{\chi}^M)$ correspond to the case where no seed is added, and $(\bar{\bm{\chi}}^1,\ldots,\bar{\bm{\chi}}^M)$ correspond to the case where $K_m$ seeds are added to $t_m$ for each $m=1,\ldots,M$.
The outline of the proof is that, we first show that a) the total variation distance of the two distributions $(\bm{\chi}^1,\ldots,\bm{\chi}^M)$ and $(\bar{\bm{\chi}}^1,\ldots,\bar{\bm{\chi}}^M)$ is $o(1)$; then b) we show that $\sigma(\bm{\chi}^1,\ldots,\bm{\chi}^M)$ and $\sigma(\bar{\bm{\chi}}^1,\ldots,\bar{\bm{\chi}}^M)$ can only differ by $o(n)$ in expectation.

We first note that claim a) can imply claim b) easily.
Notice that the range of the function $\sigma(\cdot)$ falls into the interval $[0,n]$.
The total variation distance of $(\bm{\chi}^1,\ldots,\bm{\chi}^M)$ and $(\bar{\bm{\chi}}^1,\ldots,\bar{\bm{\chi}}^M)$ being $o(1)$ implies that
$$\left|\E_{(\bm{\chi}^1,\ldots,\bm{\chi}^M)}[\sigma(\bm{\chi}^1,\ldots,\bm{\chi}^M)]-\E_{(\bar{\bm{\chi}}^1,\ldots,\bar{\bm{\chi}}^M)}[\sigma(\bar{\bm{\chi}}^1,\ldots,\bar{\bm{\chi}}^M)]\right|=o(n),$$
by a standard property of total variation distance (see, for example, Proposition~4.5 in~\cite{levin2017markov}).

To show the claim a), noticing that $M$ is a constant and $\bm{\chi}^{m_1}$ is independent of $\bm{\chi}^{m_2}$ for any $m_1$ and $m_2$ (the appearances of edges between $V(t_{m_1})$ and $V\setminus\bigcup_{m=1}^MV(t_m)$ are independent of the appearances of edges between $V(t_{m_2})$ and $V\setminus\bigcup_{m=1}^MV(t_m)$), it is sufficient to show that the total variation distance between $\bm{\chi}^m$ and $\bar{\bm{\chi}}^m$ is $o(1)$.
Each vertex $v\in V(t_m)$ is connected to an arbitrary vertex in a critical or supercritical leaf with probability between $\omega(1/n^{1+1/r})$ (since the root has weight $\omega(1/n^{1+1/r})$) and $o(1/n)$ (otherwise $t_m$ is vulnerable).
Since the number of infected vertices in $V\setminus\bigcup_{m=1}^MV(t_m)$ is $\Theta(n)$, the number of $v$'s infected neighbors follows a Binomial distribution, $\mbox{Bin}(n,\theta)$, with mean $n\theta$ between $\omega(1/n^{1/r})$ and $o(1)$, we can use Poisson distribution $\po(n\theta)$ to approximate it.  Formally, the total variation distance is $d_{TV}(\mbox{Bin}(n,\theta), \po(n\theta)) \le n \theta^2 = o(1/n)$.
Thus, this approximation only changes the total variation distance of $\bm{\chi}^m$ by $o(1)$.
Observing this, the proposition below shows the total variation distance between $\bm{\chi}^m$ and $\bar{\bm{\chi}}^m$ is $o(1)$.

\begin{proposition}
Let $\lambda$ be such that $\lambda=\omega(1/n^{1/r})$ and $\lambda = o(1)$.
Let $Y_1,\ldots,Y_n\in\mathbb{Z}$ be $n$ independently and identically distributed random variables where each $Y_i$ is sampled from a Poisson distribution with mean $\lambda$.
Let $Z_1,\ldots,Z_n\in\mathbb{Z}$ be $n$ random variables, where the first $K$ of them satisfy $Z_1=\cdots=Z_K=r$ with probability $1$, and the remaining random variables $Z_{K+1},\ldots,Z_n$ are independently sampled from a Poisson distribution with mean $\lambda$.
For $i=0,1,\ldots,r-1$, let $\chi_i$ be the number of random variables in $\{Y_1,\ldots,Y_n\}$ that have value $i$, and $\bar{\chi}_i$ be the number of random variables in $\{Z_1,\ldots,Z_n\}$ that have value $i$.
Let $\chi_r$ be the number of random variables in $\{Y_1,\ldots,Y_n\}$ that have values at least $r$, and $\bar{\chi}_r$ be the number of random variables in $\{Z_1,\ldots,Z_n\}$ that have values at least $r$.
The total variation distance between $\bm\chi=(\chi_0,\chi_1,\ldots,\chi_r)$ and $\bar{\bm{\chi}}=(\bar{\chi}_0,\bar{\chi}_1,\ldots,\bar{\chi}_r)$ is $d_{TV}(\bm{\chi}, \bar{\bm{\chi}}) = o(1)$ if $K=\Theta(1)$.
\end{proposition}
To show that random vectors $\bm\chi$ and $\bar{\bm\chi}$ have a small total variation distance, we first estimate them by Poisson approximations.  Note that $\bm{\chi}$ and $\bar{\bm{\chi}}$ can be seen as ball and bin processes.  There are $r+1$ bins, and $n$ balls.  For $\bm{\chi}$, the probability of ball $i$ in bin $\ell$ is $\Pr[Y_i = \ell]$ when $0\le \ell<r$ and $\Pr[Y_i \ge r]$ for bin $r$.  $\chi_\ell$ is the number of balls in bin $\ell$.
Therefore, we can simplify the correlation between the coordinates of $\bm{\chi} = (\chi_0, \chi_1, \ldots, \chi_r)$, and formulate $\bm{\chi}$ as a $r+1$ coordinate-wise independent Poisson $\bm{\zeta} = (\zeta_0, \zeta_1, \ldots, \zeta_r)$  with the same expectation $\E[\bm{\chi}] = \E[\bm{\zeta}]$ conditioning on $\sum_{0\le \ell\le r} \zeta_\ell = n$.  For $\bar{\bm{\chi}}$, we define $\bar{\bm{\zeta}}$ similarly.

Then, we upper-bound the total variation distance between those two Poisson vectors $\bm{\zeta}$ and $\bar{\bm{\zeta}}$ conditioning on $\sum_{0\le \ell\le r} \zeta_\ell = \sum_{0\le \ell\le r} \bar{\zeta}_\ell = n$.  We compute the relative divergence between them and use the Pinsker's inequality~\cite{kemperman1969optimum} to upper bound the total variation distance.

\begin{proof}
For $\bm{\chi}$, there are $r+1$ bins and $n$ balls.
Let the probability of ball $i$ in bin $\ell$ be $p_\ell := \Pr[Y_i = \ell]$ when $0\le \ell<r$ and $p_r:= \Pr[Y_i \ge r]$ for bin $r$ (note that these probabilities are independent of the index $i$).  For $0\le \ell\le r$, $\chi_\ell$ is the number of balls in bin $\ell$.
Consider the following Poisson vector $\bm{\zeta} = (\zeta_0, \zeta_1, \ldots, \zeta_{r})$ with parameters $(\lambda_0, \ldots, \lambda_r)$ where $\lambda_\ell = np_\ell$ for $0\le \ell\le r$: each coordinate $\zeta_\ell$ is sampled from a Poisson distribution with parameter $\lambda_\ell$ independently.  Note that the distribution of $\bm{\chi}$ equals to $\bm{\zeta}$ conditioning on $\sum_{\ell=0}^r \zeta_\ell = n$: for all $\mathbf{k} \in \mathbb{Z}_{\geq0}^{r+1}$ with $\sum_{\ell=0}^r k_\ell = n$,
\begin{equation}\label{eq:poisson1}
    \Pr\left(\bm{\chi} = \mathbf{k}\right) = \Pr\left(\bm{\zeta} = \mathbf{k}\mid \sum_{\ell=0}^r k_\ell = n\right) = \frac{n!}{n^ne^{-n}}\prod_{\ell=0}^{r}\frac{\lambda_\ell^{k_\ell} e^{-\lambda_\ell}}{k_\ell!}.
\end{equation}

The process $\bar{\bm\chi}$ needs more work.
In the context of ball and bin process, the first $K$ balls are in bin $r$ with probability $1$, and the rest of balls follow the distribution $(p_\ell)_{0\le \ell\le r}$ defined above.   For $0\le \ell\le r$, $\bar{\chi}_\ell$ is the number of balls in bin $\ell$.
This non-symmetry makes the connection from $\bar{\bm{\chi}}$ to a Poisson distribution less obvious.
Here, we first use a process $\bar{\bm{\chi}}'$ to approximate $\bar{\bm{\chi}}$ where all balls are thrown into the bins independently and identically, and we translate $\bar{\bm{\chi}}'$ to a Poisson distribution.
Before defining $\bar{\bm{\chi}}'$, note that $\bar{\bm{\chi}}$ is equivalent to the following process: instead of picking first $K$ indices, we can randomly pick $K$ indices $i_1, i_2, \ldots, i_K$ and let $Z_{i_{\iota}} = r$ for $0\le \iota\le K$. The other follows the distribution $(p_\ell)_{0\le \ell\le r}$.
In this formulation, the distribution of the positions of balls are identical, but not independent.
Now we define $\bar{\bm{\chi}}'$ by setting them to be independent: Let the probability of ball $i$ in bin $\ell$ be $\bar{p}_\ell:= (1-K/n)p_\ell$ when $0\le \ell<r$ and $\bar{p}_r:= (1-K/n)p_r+K/n$.
The positions of balls are now mutually independent in $\bar{\bm{\chi}}'$.
For $0\le \ell\le r$, $\bar{\chi}_\ell'$ is the number of balls in bin $\ell$.

Note that the distributions of $\bar{\bm{\chi}}$ and $\bar{\bm{\chi}}'$ are different.
In particular, the marginal distribution of $\bar{\chi}_r$ is $K$ plus a binomial distribution with parameter $(n-K, p_r)$, and the marginal distribution of $\bar{\chi}_r'$ is a binomial distribution with parameter $(n, \bar{p}_r)$.
However, we can show that \begin{equation}\label{eq:poisson3}
    d_{TV}(\bar{\bm{\chi}}, \bar{\bm{\chi}}') = o(1).
\end{equation}
Equivalently, we want to show there exists a coupling between $\bar{\bm{\chi}}$ and $\bar{\bm{\chi}}'$ such that the probability of $\bar{\bm{\chi}}\neq \bar{\bm{\chi}}'$ is in $o(1)$.
First, for all $k_r\ge K$, the distributions of $\bar{\bm{\chi}}$ conditioning on $\bar{\chi}_r = k_r$ and $\bar{\bm{\chi}}'$ conditioning on $\bar{\chi}_r' = k_r$ are the same.
Therefore, fixing a coupling between $\bar{\chi}_r$ and $\bar{\chi}_r'$, we can extend it to a coupling between $\bar{\bm{\chi}}$ and $\bar{\bm{\chi}}'$ such that when an event $\bar{\chi}_r = \bar{\chi}_r'$ happens, $\bar{\bm{\chi}} = \bar{\bm{\chi}}'$.
Thus, we have $d_{TV}(\bar{\bm{\chi}}, \bar{\bm{\chi}}') = d_{TV}(\bar{\chi}_r, \bar{\chi}_r')$.  Now it suffices to show the following claim.
\begin{claim}\label{claim:poission1}
$$d_{TV}(\bar{\chi}_r, \bar{\chi}_r') = o(1).$$
\end{claim}
Intuitively, the mean of $\bar{\chi}_r$ and $\bar{\chi}_r'$ are both $n\bar{p}_r$ which is in $\omega(1)$, so the small distinction between them should not matter.
We present a proof later for completeness.

Given $\bar{\bm{\chi}}'$, consider the following Poisson vector $\bar{\bm{\zeta}} = (\bar{\zeta}_0, \bar{\zeta}_1, \ldots, \bar{\zeta}_{r})$ with parameter $(\bar{\lambda}_0, \ldots, \bar{\lambda}_r)$ where $\bar{\lambda}_\ell = n\bar{p}_\ell$ for $0\le \ell\le r$.
The distribution of $\bar{\bm{\chi}}'$ equals to $\bar{\bm{\zeta}}$ conditioning on $\sum_{\ell=0}^r \bar{\zeta}_\ell = n$: for all $\mathbf{k} \in \mathbb{Z}_{\geq0}^{r+1}$ with $\sum_{\ell=0}^r k_\ell = n$,
\begin{equation}\label{eq:poisson2}
    \Pr\left(\bar{\bm{\chi}}' = \mathbf{k}\right) = \Pr\left(\bar{\bm{\zeta}} = \mathbf{k}\mid \sum_{\ell=0}^r k_\ell = n\right) = \frac{n!}{n^ne^{-n}}\prod_{\ell=0}^{r}\frac{\bar{\lambda}_\ell^{k_\ell} e^{-\bar{\lambda}_\ell}}{k_\ell!}.
\end{equation}

Finally, with \eqref{eq:poisson1}, \eqref{eq:poisson3}, and \eqref{eq:poisson2}, it suffices to upper-bound the total variation distance between $\bm{\chi}$ and $\bar{\bm{\chi}}'$.
We will prove the following claim later.
\begin{claim}\label{claim:poission2}
$$d_{TV}(\bm{\chi}, \bar{\bm{\chi}}') = o(1).$$
\end{claim}
With these claims, we completes the proof:
$$d_{TV}(\bm{\chi}, \bar{\bm{\chi}})\le d_{TV}(\bm{\chi}, \bar{\bm{\chi}}')+d_{TV}(\bar{\bm{\chi}}, \bar{\bm{\chi}}') = o(1)$$
by the triangle inequality.
\end{proof}
\begin{proof}[Proof of Claim~\ref{claim:poission1}]
Informally, the mean of $\bar{\chi}_r$ and $\bar{\chi}_r'$ are both $n\bar{p}_r = \omega(1)$, so the small distinction between them should not matter.  We formalize these by using Poisson distributions to approximate $\bar{\chi}_r$ (a binomial, $\mbox{Bin}(n, \bar{p}_r)$) and $\bar{\chi}_r'$ (a transported binomial, $K+\mbox{Bin}(n-K, p_r)$).

Recall that $\po(x)$ denotes a Poisson random variable with parameter $x$.
By the triangle inequality, the distance, $d_{TV}(\bar{\chi}_r, \bar{\chi}_r') = d_{TV}(K+\mbox{Bin}(n-K, p_r), \mbox{Bin}(n, \bar{p}_r))$, is less the the sum of the following four terms:
\begin{enumerate}
    \item $d_{TV}(K+\mbox{Bin}(n-K, p_r), K+\po((n-K)p_r))$,
    \item $d_{TV}(K+\po((n-K)p_r),K+ \po(n\bar{p}_r))$,
    \item $d_{TV}(K+ \po(n\bar{p}_r), \po(n\bar{p}_r))$, and
    \item $d_{TV}(\po(n\bar{p}_r), \mbox{Bin}(n, \bar{p}_r))$.
\end{enumerate}
Now we want to show all four terms are in $o(1)$.
By the Poisson approximation~\cite{le1960approximation}, for all $p$, $d_{TV}(\mbox{Bin}(n,p), \po(np))\le p$, the first and the final term, are less than $p_r$ and $\bar{p}_r$ respectively.
Both are in $o(1)$ since $p_r = \Theta(\lambda^r)$.

For the second term, because $d_{TV}(\po(\lambda_1), \po(\lambda_2))\le \frac{|\lambda_1-\lambda_2|}{\sqrt{\lambda_1}+\sqrt{\lambda_2}}$ for all $\lambda_1$ and $\lambda_2$ (see~\cite{adell2006exact}) and $p_r = \Omega(\lambda^r) = \omega(1/n)$,
\begin{align*}
    d_{TV}(\po((n-K)p_r), \po(n\bar{p}_r))
    \le \frac{n\bar{p}_r-(n-K)p_r}{\sqrt{n\bar{p}_r}+\sqrt{(n-K)p_r}} =  \frac{K+Kp_r}{\sqrt{n\bar{p}_r}+\sqrt{(n-K)p_r}} = o(1).
\end{align*}
Finally, for the third term, let $(x)^+ = \max\{0,x\}$ for all $x$.
Recall that $\bar{\lambda}_r = n\bar{p}_r$.
By a definition of total variation distance, we have
\begin{align*}
    &d_{TV}(K+ \po(\bar{\lambda}_r), \po(\bar{\lambda}_r))\\
    =& \sum_{x \ge 0} \left(\Pr(K+ \po(\bar{\lambda}_r) = x)-\Pr(\po(\bar{\lambda}_r) = x)\right)^+\\
    =& \sum_{x \ge K} \left(\Pr(\po(\bar{\lambda}_r) = x-K)-\Pr(\po(\bar{\lambda}_r) = x)\right)^+\tag{the first $K$ terms are zero}\\
    =& \sum_{x \ge K} \Pr(\po(\bar{\lambda}_r) = x-K)\left(1-\frac{\Pr(\po(\bar{\lambda}_r) = x)}{\Pr(\po(\bar{\lambda}_r) = x-K)}\right)^+\\
    =& \sum_{x \ge 0} \Pr(\po(\bar{\lambda}_r) = x)\left(1-\frac{\Pr(\po(\bar{\lambda}_r) = x+K)}{\Pr(\po(\bar{\lambda}_r) = x)}\right)^+\tag{change variable}\\
    =& \sum_{x \ge 0} \Pr(\po(\bar{\lambda}_r) = x)\left(1-\frac{(\bar{\lambda}_r)^K}{(x+1)(x+2)\ldots (x+K)}\right)^+
\end{align*}
Because ${(x+1)(x+2)\ldots (x+K)}$ is increasing as $x$ increases, there exists $x^*$ such that $(\bar{\lambda}_r)^K\le (x+1)(x+2)\ldots (x+K)$ if and only if $x\ge x^*$.
Therefore,
\begin{align*}
    &d_{TV}(K+ \po(\bar{\lambda}_r), \po(\bar{\lambda}_r))\\
   =& \sum_{x \ge 0} \Pr(\po(\bar{\lambda}_r) = x)\left(1-\frac{(\bar{\lambda}_r)^K}{(x+1)(x+2)\ldots (x+K)}\right)^+\\
   =& \sum_{x \ge x^*} \Pr(\po(\bar{\lambda}_r) = x)\left(1-\frac{(\bar{\lambda}_r)^K}{(x+1)(x+2)\ldots (x+K)}\right)\\
   =& \Pr(\po(\bar{\lambda}_r) \ge x^*)-\Pr(\po(\bar{\lambda}_r) \ge x^*+K)\\
   =& \sum_{x = x^*}^{x^*+K-1}\Pr(\po(\bar{\lambda}_r) = x)\le K\max_x\Pr(\po(\bar{\lambda}_r) = x)
\end{align*}
Now we want to show $\max_x\Pr(\po(\bar{\lambda}_r) = x) = o(1)$.  Intuitively, since the expectation $\bar{\lambda}_r = \omega(1)$ is large, the probability mass function $\Pr(\po(\bar{\lambda}_r) = x)$ is ``flat'', and the maximum of the probability mass function is small.  Formally, for all $x$, $\Pr(\po(\bar{\lambda}_r) = x+1)/\Pr(\po(\bar{\lambda}_r) = x) = \bar{\lambda}_r/(x+1)$, so the maximum happens at $x_M := \lfloor \bar{\lambda}_r\rfloor$.Then we can compute an upper bound of $\Pr(\po(\bar{\lambda}_r) = x_M)$ by Stirling approximations.
\begin{align*}
&\Pr\left(\po(\bar{\lambda}_r) = x_M\right)= \frac{(\bar{\lambda}_r)^{x_M}e^{-\bar{\lambda}_r}}{x_M!}\le \frac{(\bar{\lambda}_r)^{x_M}e^{-\bar{\lambda}_r}}{\sqrt{2\pi} x_M^{x_M+1/2}e^{-x_M}}\tag{Stirling's approximation~\cite{feller2008introduction}}\\
=& \frac{1}{\sqrt{2\pi} x_M^{1/2}}\cdot \frac{e^{-\bar{\lambda}_r}}{e^{-x_M}}\cdot \left(\frac{\bar{\lambda}_r}{x_M}\right)^{x_M}\le \frac{1}{\sqrt{2\pi x_M}}\cdot \left(\frac{\bar{\lambda}_r}{x_M}\right)^{x_M}\tag{$\bar{\lambda}_r\ge x_M$}\\
\le& \frac{1}{\sqrt{2\pi x_M}}\cdot \left(1+\frac{\bar{\lambda}_r-x_M}{x_M}\right)^{x_M}
\le \frac{1}{\sqrt{2\pi x_M}}\cdot \left(1+\frac{1}{x_M}\right)^{x_M}\le \frac{e}{\sqrt{2\pi x_M}} = o(1)
\end{align*}
The last one holds because $x_M = \lfloor \bar{\lambda}_r\rfloor = \omega(1)$.
\end{proof}
\begin{proof}[Proof of Claim~\ref{claim:poission2}]
Because the distributions of $\bm{\chi}$ and $\bar{\bm{\chi}}'$ are very close to product distributions, the relative entropy between them is easier to compute than the total variation distance.  By Pinsker's inequality, if the relative entropy is small, the total variation distance is also small.
\begin{align*}
    D_{KL}(\bm{\chi}\|\bar{\bm{\chi}}') =& -\sum_{\mathbf{k}:\sum_{\ell=0}^r k_\ell = n} \Pr(\bm{\chi} = \mathbf{k}) \log \frac{\Pr(\bar{\bm{\chi}}' = \mathbf{k})}{\Pr(\bm{\chi} = \mathbf{k})}\\
    =& -\sum_{\mathbf{k}:\sum_{\ell=0}^r k_\ell = n} \Pr(\bm{\chi} = \mathbf{k}) \log \left(\frac{\prod_{\ell = 0}^{r}\frac{\bar{\lambda}_\ell^{k_\ell} e^{-\bar{\lambda}_\ell}}{k_\ell!}}{\prod_{\ell = 0}^{r}\frac{\lambda_\ell^{k_\ell} e^{-\lambda_\ell}}{k_\ell!}}\right)\tag{by Eqn.~\eqref{eq:poisson1} and \eqref{eq:poisson2}}\\
    =& -\sum_{\mathbf{k}:\sum_{\ell=0}^r k_\ell = n} \Pr(\bm{\chi} = \mathbf{k}) \left(\sum_{\ell = 0}^{r}k_\ell \log \frac{\bar{\lambda}_\ell}{\lambda_\ell}\right)\tag{because $\sum_{\ell=0}^r \lambda_\ell = \sum_{\ell=0}^r \bar{\lambda}_\ell$}\\
    =& -\sum_{\mathbf{k}:\sum_{\ell=0}^r k_\ell = n} \Pr(\bm{\chi} = \mathbf{k}) \left( \sum_{\ell = 0}^{r-1}k_\ell\log\left(1-\frac{K}{n}\right)+k_r\log\left(1+(1/p_r-1)\frac{K}{n}\right) \right)
\end{align*}
In the outermost parentheses, everything except $k_\ell$ and $k_r$ are independent of the summation over $\mathbf{k}$, so we can simplify it as the following:
\begin{align*}
    D_{KL}(\bm{\chi}\|\bar{\bm{\chi}}') =& -\left[\log\left(1-\frac{K}{n}\right) \sum_{\mathbf{k}} \Pr(\bm{\chi} = \mathbf{k}) \left(\sum_{\ell = 0}^{r-1}k_\ell\right)+\log\left(1+(1/p_r-1)\frac{K}{n}\right) \sum_{\mathbf{k}} \Pr(\bm{\chi} = \mathbf{k})k_r\right]\\
    =& -\left[\log\left(1-\frac{K}{n}\right) \sum_{\ell = 0}^{r-1}\E[\chi_\ell]+\log\left(1+(1/p_r-1)\frac{K}{n}\right) \E[\chi_r]\right]\\
    =& -\left[(n-\lambda_r)\log\left(1-\frac{K}{n}\right)  +\lambda_r \log\left(1+(1/p_r-1)\frac{K}{n}\right) \right]\tag{$\E[\chi_\ell] = \lambda_\ell$}\\
    =& -n\left[(1-p_r)\log\left(1-\frac{K}{n}\right)  +p_r \log\left(1+(1/p_r-1)\frac{K}{n}\right) \right]\tag{$\lambda_r = np_r$}
\end{align*}
Now we want to show $(1-p_r)\log\left(1-\frac{K}{n}\right)  +p_r \log\left(1+(1/p_r-1)\frac{K}{n}\right)$ is $o(1/n)$.  Because $p_r = \Pr(Y_i\ge r) = \Theta(\lambda^r) = \omega(1/n)$ and $K$ is a constant, we can use Taylor expansion to approximate both logs at $1$,
\begin{align*}
 &(1-p_r)\log\left(1-\frac{K}{n}\right)  +p_r \log\left(1+(1/p_r-1)\frac{K}{n}\right)\\
 =& -(1-p_r)\frac{K}{n}+p_r (1/p_r-1)\frac{K}{n}+O\left(\frac{1}{n^2}\right)+O\left(\frac{1}{p_r n^2}\right)\\
 =& O\left(1/(p_r n^2)\right) = o(1/n)\tag{because $p_rn = \omega(1)$}
\end{align*}
Therefore, we have $D_{KL}(\bm{\chi}\|\bar{\bm{\chi}}') = o(1)$.  By Pinsker's inequality
$$d_{TV}(\bm{\chi}, \bar{\bm{\chi}}')\le \sqrt{\frac{1}{2}D_{KL}(\bm{\chi}\| \bar{\bm{\chi}}')} = o(1).$$
\end{proof}

\section{Proof of Proposition~\ref{prop:isolate}}
By Theorem~\ref{thm:gnp_sub} and Corollary~\ref{cor:gnp_boundary}, if no leaf is activated by the local seeds, then there can be at most constantly many infected vertices with high probability.
Consider an arbitrary vertex $v$ that is not infected, and let $t$ be the leaf such that $v\in V(t)$.
Let $K_{in}$ be the number of infected vertices in $V(t)$ after Phase I and $K_{out}$ be the number of infected vertices outside $V(t)$.
By our assumption, $K_{in}=O(1)$ and $K_{out}=O(1)$.
We compute an upper bound on the probability that $v$ is infected in the next cascade iteration.
Let $X_v$ be the number of $v$'s infected neighbors in $V(t)$ and $Y_v$ be the number of $v$'s infected neighbors outside $V(t)$.

Since the probability that $v$ is connected to each of those $K_{out}$ vertices is $o(n^{-1/r})$, we have
$$\Pr(Y_v\geq r-a)\leq\binom{K_{out}}{r-a}\left(o(n^{-1/r})\right)^{r-a}=o\left(n^{-(r-a)/r}\right)$$
for each $a\in\{0,1,\ldots,r-1\}$.

Ideally, we would also like to claim that
\begin{equation}\label{eqn:dom}
\Pr(X_v\geq a)\leq\binom{K_{in}}{a}w(t)^a=O\left(n^{-a/r}\right),
\end{equation}
so that putting together we have,
$$\Pr(v\mbox{ is infected})\leq\sum_{a=0}^{r-1}\Pr(X_v\geq a)\Pr(Y_v\geq r-a)=r\cdot O\left(n^{-a/r}\right)\cdot o\left(n^{-(r-a)/r}\right)=o\left(\frac1n\right).$$
and conclude that the expected number of infected vertices in the next iteration is $o(1)$, which implies the proposition by the Markov's inequality.

However, conditioning on the cascade in $V(t)$ stopping after $K_{in}$ infections, there is no guarantee that the probability an edge between $v$ and one of the $K_{in}$ infected vertices is still $w(t)$.
Moreover, for any two vertices $u_1,u_2$ that belong to those $K_{in}$ infected vertices, we do not even know if the probability that $v$ connects to $u_1$ is still independent of the probability that $v$ connects to $u_2$.
Therefore, (\ref{eqn:dom}) does not hold in a straightforward way.
The remaining part of this proof is dedicated to prove (\ref{eqn:dom}).

Consider a different scenario where we have put $K_{in}$ seeds in $V(t)$ (instead of that the cascade in $V(t)$ ends at $K_{in}$ infections), and let $\bar{X}_v$ be the number of edges between $v$ and those $K_{in}$ seeds (where $v$ is not one of those seeds).
Then we know each edge appears with probability $w(t)$ independently, and (\ref{eqn:dom}) holds for $\bar{X}_v$:
$$\Pr(\bar{X}_v\geq a)\leq\binom{K_{in}}{a}w(t)^a=O\left(n^{-a/r}\right).$$

Finally, (\ref{eqn:dom}) follows from that $\bar{X}_v$ stochastically dominates $X_v$ (i.e., $\Pr(\bar{X}_v\geq a)\geq\Pr(X_v\geq a)$ for each $a\in\{0,1,\ldots,r-1\}$), which is easy to see:
$$\Pr\left(X_v\geq a\right)=\Pr\left(\bar{X}_v\geq a\mid \bar{X}_v\leq r-1\right)=\frac{\Pr(a\leq \bar{X}_v\leq r-1)}{\Pr(\bar{X}_v\leq r-1)}$$
$$\qquad=\frac{\Pr(\bar{X}_v\geq a)-\Pr(\bar{X}_v\geq r)}{1-\Pr(\bar{X}_v\geq r)}\leq\Pr\left(\bar{X}_v\geq a\right),$$
where the first equality holds as $\Pr\left(\bar{X}_v\geq a\mid \bar{X}_v\leq r-1\right)$ exactly describes the probability that $v$ has at least $a$ infected neighbors among $K_{in}$ conditioning on $v$ has not yet been infected.
\end{document}